\documentclass[journal]{IEEEtran}

\usepackage[T1]{fontenc}

\usepackage{xcolor}
\usepackage{soul} 
\usepackage{subcaption}
\usepackage{listings}
\usepackage{geometry}
\usepackage{setspace}
\usepackage{booktabs}
\usepackage{makecell}
\usepackage{siunitx}
\usepackage{cuted}
\usepackage{biblatex}
\usepackage{graphicx}
\usepackage{float}
\newfloat{scheme}{htbp}{los}
\floatname{scheme}{Scheme}
\floatname{chart}{Chart}
\newfloat{graph}{htbp}{loh}

\usepackage{chemformula} 
\usepackage[version = 4]{mhchem} 

\usepackage{amsmath}
\usepackage{amssymb}
\usepackage{url}
\usepackage{tabularx}
\usepackage{xurl}
\usepackage{amsthm}
\newtheorem*{lemma}{Lemma}
\newtheorem*{claim}{Claim}

\newcommand{\mycomment}[1]{}

\usepackage{authblk}
\usepackage{hyperref}
\author[1]{Hamideh Bakhshi}
\author[1]{Hamid Abdollahi*}
\affil[1]{Department of Chemistry, Institute for Advanced Studies in Basic Sciences (IASBS), 444 Prof Yousef Sobouti Blvd, Zanjan, 45137-66731, Iran}
\author[2]{R\'{o}bert Rajk\'{o}*}
\affil[2]{University Research and Innovation Center (EKIK), University of \'{O}buda, B\'{e}csi \'{u}t 96/b, Budapest, H-1034, Hungary}

\title{Some surprising properties of \\ essential data points visualization}
\date{*Emails: abd@iasbs.ac.ir, rajko.robert@uni-obuda.hu}

\bibliography{ieee-template.bib}

\definecolor{codebg}{rgb}{0.96,0.96,0.96}
\definecolor{codecomment}{rgb}{0.13,0.55,0.13}
\definecolor{codekeyword}{rgb}{0.0,0.0,0.65}
\definecolor{codestring}{rgb}{0.63,0.13,0.13}

\lstdefinestyle{matlabstyle}{
	language=Matlab,
	backgroundcolor=\color{codebg},
	basicstyle=\footnotesize\ttfamily,
	keywordstyle=\color{codekeyword}\bfseries,
	commentstyle=\color{codecomment}\itshape,
	stringstyle=\color{codestring},
	numbers=left,
	numberstyle=\tiny\color{gray},
	stepnumber=1,
	numbersep=8pt,
	breaklines=true,
	breakatwhitespace=false,
	frame=single,
	rulecolor=\color{black!30},
	tabsize=2,
	showstringspaces=false,
	captionpos=b,
	linewidth=\dimexpr\textwidth-2\fboxsep\relax,
}

\begin{document}

\maketitle

\begin{abstract}
Essential Data Points (EDPs) --- the vertices of the convex hull of a bilinear
data matrix $D$ in its row space, column space, or both --- are widely used in
chemometrics to reduce the size of large data sets while nominally preserving
their underlying geometric structure. Using simulated three- and
two-component chromatographic/spectral data sets and a real
source-apportionment data set ($\mathbf{D} = \mathbf{C}\, \mathbf{A}^\mathsf{T}$), together with
Borgen--Rajk\'o plots, Procrustes analysis, and variance--covariance
comparisons, we show that this preservation is only \emph{partial}: row-wise
EDP reduction preserves the row-space geometry (inner and outer polygons)
exactly while distorting the column-space geometry, and column-wise reduction
shows the opposite behavior; joint row-and-column reduction distorts both. We
then give a rigorous, general proof --- based on the four fundamental
subspaces of a matrix and its singular value decomposition $\mathbf{D}=\mathbf{U}\,\mathbf{S}\,\mathbf{V}^\mathsf{T}$
--- that the subspace which is \emph{not} being reduced is always preserved
exactly, up to an orthogonal rotation, whereas the subspace whose ambient
dimension shrinks is related to the original only through a general,
non-orthogonal isomorphism. This distinction is confirmed numerically to
machine precision ($\sim 10^{-14}$--$10^{-16}$) on the real data set, and a
deliberate negative control confirms that the "ambient-shrinking" map is
genuinely non-orthogonal (residual $\approx 1$). These results demonstrate
that the apparent rotation of an EDP-reduced polygon relative to the original
is not, in general, a rigid rotation, and that visual or numerical
comparisons between an EDP-reduced data set and the original data require an
explicit, mode-dependent change-of-basis correction before any geometric or
statistical conclusion can be drawn. A MATLAB implementation of this
correction is provided.
\end{abstract}




\section{Introduction}


The term data reduction methods refers to a collection of mathematical algorithms that are used to reduce the dimensionality and/or size of large data sets (\cite{Pearson1901,Scholkopf1998,vandermaaten2008,Ivosev2008,ayesha2020,KHODADADIKARIMVAND2023,wani2025,GARISO2025,cordina2026}). Various algorithms have been proposed for different purposes, including variable selection and sample selection, in order to eliminate redundant information. Among these approaches, methods capable of preserving the geometric structure of data points have attracted considerable attention.

An effective strategy for data reduction in bilinear data sets is the selection of Essential Data Points (EDPs) (\cite{Ghaffari2019,Ruckebusch2020,Ghaffari2021}), which represent the most informative samples and variables while preserving the underlying data structure. A common approach for identifying these points is based on the Convex Hull (CH) constructed in the normalized data space. Probably, the first use of CH in self-modeling/multivariate curve resolution (S/MCR) was in~\cite{rajko2005}.

More specifically, the convex hull of a bilinear data set is defined as the smallest convex polytope that encloses all data points so that every point lies either on the boundary or within the interior of the polytope. This polytope is formed by a unique subset of data points that create the tightest boundary that surrounds the entire data cloud and can be determined independently within each data subspace. 
In analytical chemistry/chemometrics several attempts have been made to use of CH before EDP was introduced (\cite{Fernandez2002,fernandez2003,Jin2003,Jin2003Delaunay,JIN2006,corona2010}).
The data points corresponding to the vertices of the convex hull are called EDPs, as they retain the fundamental structural information of the data set. Consequently, these points can be used as a reduced yet representative subset of samples for subsequent data analysis and modeling tasks.

Assume that $\mathbf{D}$ is the original data matrix with dimensions $m \times n$, it can be reduced in the row mode, the column mode, or jointly in both modes. Let $p$ and $q$ denote the numbers of essential data points identified in the row and column spaces, respectively. In row-wise reduction, the data matrix is transformed into a matrix of size ($\mathbf{D}_{p \times n}$), in column-wise reduction, the dimensions of the matrix become ($\mathbf{D}_{m \times q}$) and finally, when data reduction is performed jointly in both modes, the resulting matrix has dimensions of ($\mathbf{D}_{p \times q}$).

The main objective of this paper is to compare the original data matrix with the
reduced data matrices and to evaluate the structural modifications that may arise during the data reduction process. To investigate the effect of data reduction on three-component bilinear data, Borgen-Rajk\'{o} plots were used as a visualization and evaluation tool.

The Borgen-Rajk\'{o} plot can be regarded as a kind of “microscope" to examine the detailed geometry of the abstract space in a three-component system~\cite{BORGEN1985,rajko2005}. This standard approach provides fundamental insights into the microstructure of a three-component data set~\cite{rajko2017}. The Borgen-Rajk\'{o} plot consists of several regions. Due to the non-negativity constraint of the data, the abstract space is bounded by a polygon, called the outer polygon. Inside this outer polygon, another convex
polygon is located, called the inner polygon, which is defined by the convex hull of the coordinates of the response vectors of the data matrix. In this framework, the purest rows or columns of the data set correspond to points lying on the facets of the inner polygon. In addition, three admissible regions are defined in the abstract space, each representing a specific portion of the feasible region of the system. Each admissible region corresponds to the set of feasible solutions associated with one of the components of the system.

However, the main question that arises is which types of information are lost during data reduction based on the essential data points and which types of information from the original data are preserved.

The structure of the paper is following: (a) after this Introduction section, some observations follow, with which we can illustrate the basic problem of the simple use of EDPs visualization, (b) summarizing theoretical considerations using redundant, understandable well-presented form, (c) illustrating the proper transformation regarding to EDPs, (d) some conclusions end the paper.

\section{Experimental}\label{sec:exprmntl}
\subsection{Three-component dataset}

To address the question of which information is preserved and which is altered or lost when data reduction is performed using essential data points, a noise-free three-component chromatographic dataset was simulated. The corresponding elution and spectral profiles of the dataset are presented in Fig.~\ref{fig1:dataset}.

\begin{figure}
  \centering
     \includegraphics[width=0.45\textwidth]{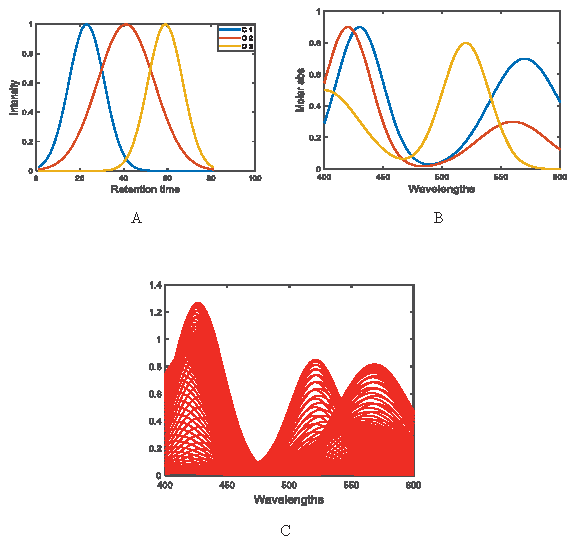} \\
  \caption{Three-component dataset. A) concentration of components in the samples, B) pure spectral profiles and C) simulated noise free data}
  \label{fig1:dataset}
 \end{figure}


The original data matrix had dimensions of $81 \times 201$. The essential points were identified by computing the convex hull vertices using MATLAB's convhull function. The analysis revealed 58 essential rows and 70 essential columns. Based on these results, three data-reduction scenarios were investigated. In the first scenario, data reduction was performed only in the row mode, resulting in a reduced matrix of size $58 \times 201$. In the second scenario, reduction was applied only in the column mode, yielding a matrix of size $81 \times 70$. Finally, jointly reduction in both row and column modes produced a reduced matrix of size $58 \times 70$. The Borgen-Rajk\'{o} plots corresponding to the original dataset and the three data-reduction scenarios, for both the row and column spaces, are presented in Fig.~\ref{fig2:reduced_data}.

\begin{figure}
  \centering
  \includegraphics[width=0.45\textwidth]{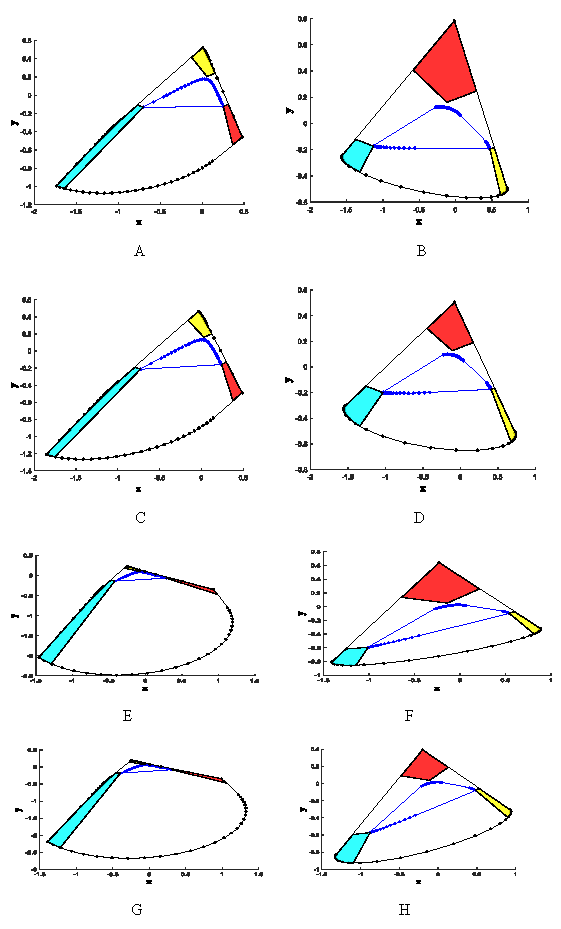} \\
  \caption{Borgen-Rajk\'{o} plots of the original and reduced datasets in the row and column spaces. (A) Original dataset in the row space; (B) original dataset in the column space; (C) row-wise reduced dataset ($\mathbf{D}_{58\times201}$) in the row space; (D) row-wise reduced dataset ($\mathbf{D}_{58\times201}$) in the column space; (E) column-wise reduced dataset ($\mathbf{D}_{81\times70}$) in the row space; (F) column-wise reduced dataset ($\mathbf{D}_{81\times70}$) in the column space; (G) jointly reduced dataset ($\mathbf{D}_{58\times70}$) in the row space; and (H) jointly reduced dataset ($\mathbf{D}_{58\times70}$) in the column space.}
  \label{fig2:reduced_data}
\end{figure}

The objective of this section is to investigate whether the geometric pattern of the essential data points is preserved after data reduction or undergoes significant changes. Since the original and reduced datasets have different dimensions, a direct comparison of their projections is not feasible. Therefore, Procrustes analysis is employed to compare their geometric configurations. Procrustes analysis is a statistical shape analysis method used to compare the geometric configuration of two or more sets of points~\cite{Carlosena1995,ANDRADE2004,kucheryavskiy2020,andreella2022,GONCALVES2023}. The main idea is that, for a meaningful comparison of “shapes,” the effects of non-shape-related factors, including translation, rotation, and scaling, must first be removed. This procedure is typically carried out in three steps:

\begin{enumerate}
  \item Translation of the points so that their centroid lies at the origin,
  \item Normalization to achieve uniform scaling,
  \item Optimal rotation to minimize the distance between the two-point configurations.
\end{enumerate}

In this study, this procedure was performed using the \verb+procrustes+ function in MATLAB.

Fig.~\ref{fig3:procrustes} presents a simplified representation of the Borgen-Rajk\'{o} plots, consisting only of the inner and outer polygons, for the original dataset and all reduced-data scenarios superimposed on the same graph. This representation enables a direct comparison of the geometric structure, shape, and relative positions of the inner and outer polygons between the original and reduced datasets.

\begin{figure}
  \centering
  \includegraphics[width=0.45\textwidth]{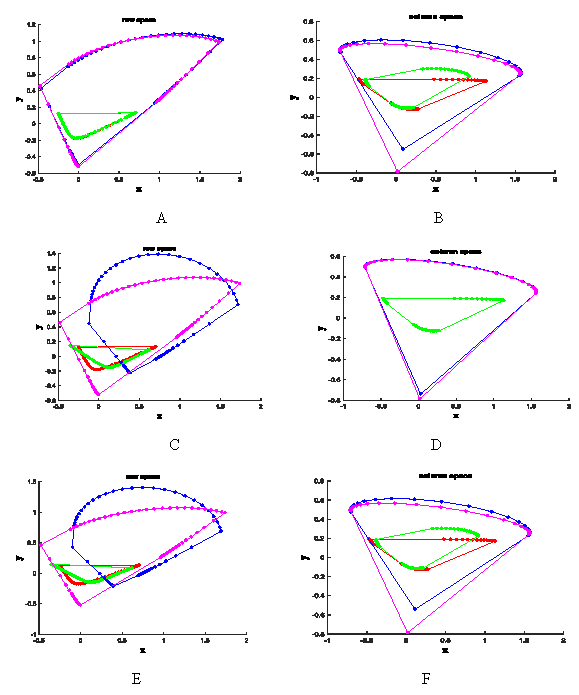} \\
  \caption{Superimposed inner and outer polygons of the original and reduced datasets after Procrustes transformation. In all panels, the inner and outer polygons of the original dataset were used as the reference structures, and the corresponding polygons of the reduced datasets were aligned using Procrustes transformation to facilitate geometric comparison. The inner polygon of the original data is shown in red, the outer polygon of the original data in pink, the inner polygon of the reduced data after Procrustes transformation in green, and the outer polygon of the reduced data after Procrustes transformation in blue. The first row corresponds to the row-wise reduced dataset ($\mathbf{D}_{58\times201}$): A) row space and B) column space, the second row corresponds to the column-wise reduced dataset ($\mathbf{D}_{81\times70}$): C) row space and D) column space, and the third row corresponds to the jointly reduced dataset ($\mathbf{D}_{58\times70}$): E) row space and F) column space.}
  \label{fig3:procrustes}
\end{figure}

As can be seen in Fig~\ref{fig3:procrustes}, for some data-reduction scenarios, the inner and outer polygons of the reduced datasets closely coincide with their corresponding polygons in the original dataset after Procrustes transformation. In contrast, for the other data reduction scenarios, the inner and outer polygons of the original and reduced datasets do not coincide completely, even after Procrustes alignment. Since Procrustes transformation removes the effects of translation, rotation, and scaling, the remaining discrepancies can be attributed to intrinsic differences in the geometric structure of the polygons. In the following sections, the effects of data reduction on the inner polygons, outer polygons, and AFS values will be examined separately and in greater detail.

\subsubsection{Changes in the inner polygons}

The first row of Fig~\ref{fig3:procrustes} corresponds to row-wise data reduction ($\mathbf{D}_{58\times201}$). As can be seen in Fig~\ref{fig3:procrustes}A, the inner and outer polygons of the reduced dataset are perfectly superimposed on those of the original dataset in the row space. In contrast, in the column space (Fig~\ref{fig3:procrustes}B), the inner and outer polygons of the original and reduced datasets do not fully coincide. The inner polygon represents the convex hull of the data points in the row or column space, and its vertices correspond to the essential points of the dataset. Therefore, the lack of coincidence between the inner polygons in the column space indicates that the geometric pattern of the essential points has changed following data reduction. The results indicate that, although data reduction does not alter the number of essential points, it changes their distribution pattern and geometric relationships. Therefore, to further investigate the effect of data reduction on the geometric structure of the dataset, the variance--covariance matrices corresponding to the essential points in the original and reduced datasets are compared and analyzed in the following section.

Since the number of essential rows is 58, the corresponding variance--covariance matrix in the row space is an ($58 \times 58$) matrix. Fig.~\ref{fig4:varcovar} presents the heat maps of the variance--covariance matrices of the essential points for the original and reduced datasets in the row space. In addition, the heat map of the residual matrix, obtained as the difference between the variance--covariance matrices of the reduced and original datasets, is also shown to facilitate the assessment of the differences between the two matrices.

\begin{figure}
  \centering
  \includegraphics[width=0.3\textwidth]{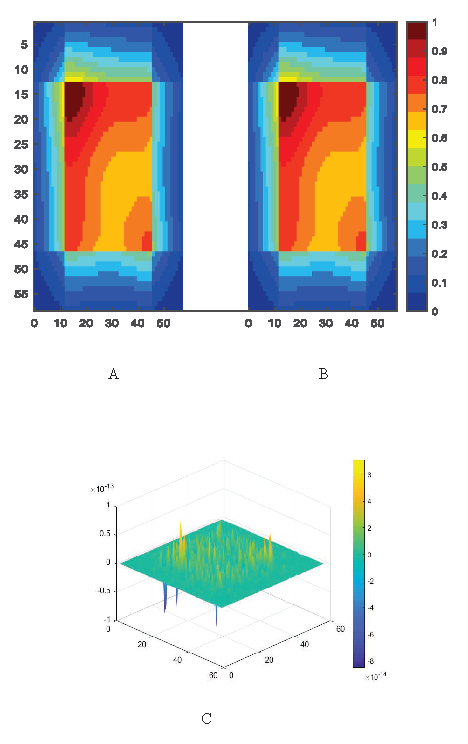} \\
  \caption{A) the variance--covariance matrices corresponding to the essential points in the original and B) the variance–covariance matrices of reduced datasets C) residual matrix obtained from the difference between the two variance--covariance matrices (variance--covariance matrices corresponding to the essential points for the original dataset and reduced data)}
  \label{fig4:varcovar}
\end{figure}

\begin{figure}
  \centering
  \includegraphics[width=0.3\textwidth]{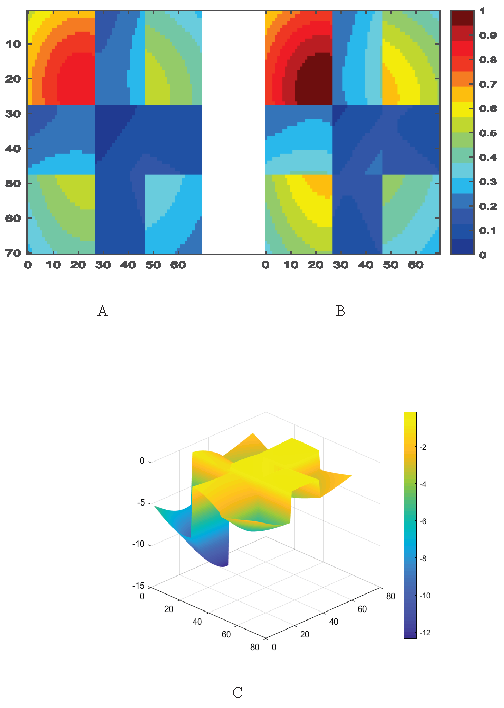} \\
  \caption{A) the variance--covariance matrices corresponding to the essential points in the original and B) the variance--covariance matrices of reduced datasets C) residual matrix obtained from the difference between the two variance–covariance matrices (variance--covariance matrices corresponding to the essential points for the original dataset and reduced
data)}
  \label{fig5:varcovar}
\end{figure}

As shown in Fig.~\ref{fig4:varcovar}, the variance--
covariance matrices corresponding to the essential points in the original and reduced datasets are identical. Furthermore, the near-zero values in the residual matrix confirm that the data reduction process does not alter the variance–covariance structure of the essential points. This result is expected because, in row-wise reduction, the selected 58 essential rows are preserved exactly in the reduced dataset, without any modification to their elements. The diagonal elements of the variance--covariance matrix correspond to the inner products of each row with itself, whereas the off-diagonal elements represent the inner products between different rows. Since the inner product of two vectors depends on their magnitudes and the cosine of the angle between them, the off-diagonal elements are directly related to the angular relationships among the rows. Therefore, the identical variance--covariance matrices indicate that the angular relationships among all essential rows are preserved after data reduction. Consequently, the geometric pattern of the essential points remains unchanged, leading to complete overlap of the inner polygons in the row space.

As shown in Fig.~\ref{fig4:varcovar}, the variance--covariance matrices corresponding to the essential points in the original and reduced datasets are identical. Furthermore, the near-zero values in the residual matrix confirm that the data reduction process does not alter the variance--covariance structure of the essential points. This result is expected because, in row-wise reduction, the selected 58 essential rows are preserved exactly in the reduced dataset, without any modification to their elements. The diagonal elements of the variance--covariance matrix correspond to the inner products of each row with itself, whereas the off-diagonal elements represent the inner products between different rows. Since the inner product of two vectors depends on their magnitudes and the cosine of the angle between them, the off-diagonal elements are directly related to the angular relationships among the rows. Therefore, the identical variance--covariance matrices indicate that the angular relationships among all essential rows are preserved after data reduction. Consequently, the geometric pattern of the essential points remains unchanged, leading to complete overlap of the inner polygons in the row space.

In contrast, the number of essential columns is 70, and the corresponding covariance matrix is an ($70 \times 70$) matrix. Fig.~\ref{fig5:varcovar} presents the heat maps of the variance–covariance matrices of the essential points for the original and reduced datasets in the column space. In addition, the residual matrix, obtained as the difference between the variance–covariance matrices of the reduced and original datasets, is also shown to facilitate the assessment of the differences between the two matrices. As can be seen, the variance--covariance matrices in the column space differ between the original and reduced datasets. This difference indicates that the geometric relationships among the essential columns have changed following data reduction. As a result, the distribution pattern of the essential points is no longer preserved in the column space, which is reflected in the fact that the inner polygons of the original and reduced datasets are not perfectly superimposed in Fig~\ref{fig3:procrustes}B.

The second row of Fig~\ref{fig3:procrustes} corresponds to column-wise data reduction ($\mathbf{D}_{81\times70}$).
As can be seen, the behavior is exactly opposite to that observed for row-wise reduction. In this case, the inner polygons of the original and reduced datasets
coincide exactly in the column space (Fig~\ref{fig3:procrustes}D), whereas noticeable differences are observed in the row space (Fig~\ref{fig3:procrustes}C). In this form of data reduction, the reduced dataset consists of a subset of the original columns that are transferred directly to the reduced matrix without modification. Consequently, the variance--covariance matrix associated with the essential columns remains identical in the original and reduced datasets. Since the geometric relationships among the essential columns are preserved, the distribution pattern of these points remains unchanged, resulting in the complete coincidence of the inner polygons in the column space. In contrast, the variance--covariance matrix associated with the essential rows differs between the original and reduced datasets. This indicates that the geometric relationships among the essential rows have been altered as a consequence of column-wise reduction. Since the off-diagonal elements of the variance--covariance matrix are directly related to the angular relationships between the rows, the observed differences imply a modification of the geometric pattern in the row space. As a result, although the essential columns and their geometric configuration are preserved, the pattern of the essential rows changes after data reduction. These observations further support the conclusion that preservation of the inner polygon depends not only on retaining the essential points themselves but also on maintaining the variance--covariance structure and geometric relationships among those points.

The third row of Fig~\ref{fig3:procrustes} corresponds to jointly data reduction in both modes ($\mathbf{D}_{58\times70}$). As can be seen, the inner polygons of the original and reduced datasets differ in both the row and column spaces. Unlike the previous two cases, where the geometric structure was preserved in one of the spaces. Based on the previous discussion, this behavior can be attributed to changes in the variance--covariance structures of both the row and column spaces. Since neither the essential rows nor the essential columns are preserved in their original geometric context, the variance--covariance matrices associated with the essential points differ between the original and reduced datasets in both spaces. These differences indicate that the geometric relationships among the essential points have been altered as a consequence of data reduction.

\subsubsection{Changes in the outer polygons}
According to singular value decomposition (SVD), the data matrix ($\mathbf{D}$) can be decomposed as $\mathbf{D} = \mathbf{U}\, \mathbf{S}\, \mathbf{V}^\text{T}$. where ($\mathbf{U}$) and ($\mathbf{V}$) contain the left and right singular vectors, respectively, and ($\mathbf{S}$) is a diagonal matrix containing the singular values.
The outer boundaries are defined based on the duality principle. In this framework, the outer boundaries in the ($\mathbf{U}$) space are estimated using ($\mathbf{V}$), while the outer boundaries in the ($\mathbf{V}$) space are determined based on ($\mathbf{U}$). The matrices ($\mathbf{U}$) and ($\mathbf{V}$) represent the basis vectors of the column space and row space, respectively. 
Similarly, the matrices ($\mathbf{C}$) and ($\mathbf{A}$) can also be interpreted as basis representations of the column and row spaces of the data. 

The first row of Fig~\ref{fig3:procrustes} corresponds to row-wise data reduction ($\mathbf{D}_{58\times201}$). In this case, the basis vectors of the row space remain exactly the same in both the original and reduced datasets. Considering the matrix factorization described by $\mathbf{D} = \mathbf{C}\,\mathbf{A}^\text{T}$. when the data matrix is reduced in the row mode, the matrix $\mathbf{A}$ remains unchanged. Consequently, the basis vectors spanning the row space are preserved, and the matrix $\mathbf{V}$, which represents the basis vectors of the row space, is identical for the original and reduced datasets. Therefore, since the matrices ($\mathbf{V}$) of the original and reduced datasets are identical, the outer polygons generated by these matrices are also identical. As a result, as shown in Fig~\ref{fig3:procrustes}A, the outer polygons of the original and reduced datasets coincide completely in the row space.

However, as shown in Fig~\ref{fig3:procrustes}B, the outer polygons of the original and reduced datasets do not coincide in the concentration space for this type of data reduction. According to $\mathbf{D} = \mathbf{C}\,\mathbf{A}^\text{T}$, the matrix ($\mathbf{C}$), which represents the basis vectors of the concentration space, is composed of dimension ($81\times3$) for the original dataset, whereas for the reduced dataset it consists of dimension ($58\times3$). To evaluate the similarity between these basis vectors, the correlation matrices corresponding to the ($\mathbf{C}$) matrices of the original and reduced datasets are reported in Table~\ref{tab:1.1}(a) and (b). This result indicates that the $\mathbf{C}$ matrix of the reduced dataset is not identical to that of the original dataset. Consequently, the $\mathbf{U}$ matrix, which defines the basis vectors of the column space, also differs between the two datasets. Since the outer polygon in the column space is determined by the basis vectors of $\mathbf{U}$, this difference leads to different outer boundaries (outer polygons) for the original and reduced datasets.

The second row of Fig~\ref{fig3:procrustes} corresponds to column-wise data reduction ($\mathbf{D}_{81\times70}$). As shown in Fig~\ref{fig3:procrustes}D, The outer polygons of the original and reduced datasets coincide completely in the concentration space. In this type of reduction, the matrix $\mathbf{C}$ of the reduced dataset is identical to that of the original dataset. Consequently, the matrix $\mathbf{U}$, which represents the basis vectors of the concentration space, is also preserved. Since the outer polygon in the concentration space is defined by the matrix $\mathbf{U}$, identical $\mathbf{U}$ matrices lead to identical outer polygons. This explains why the outer polygons of the original and reduced datasets are perfectly superimposed in Fig~\ref{fig3:procrustes}D. In contrast, the outer polygons of the original and reduced datasets differ in the spectral space, as shown in Fig~\ref{fig3:procrustes}C. The matrix $\mathbf{A}$ of the original dataset consists of three row vectors of dimension ($3\times200$), whereas the corresponding matrix $\mathbf{A}$ of the reduced dataset consists of three row vectors of dimension ($3\times70$). As a result, the basis vectors describing the spectral space are modified after data reduction. These changes propagate to the matrix $\mathbf{V}$, which represents the basis vectors of the spectral space, leading to differences between the $\mathbf{V}$ matrices of the original and reduced datasets. Consequently, the outer polygons defined by these matrices are no longer identical, resulting in the differences observed in Fig~\ref{fig3:procrustes}C.
The third row of Fig~\ref{fig3:procrustes} corresponds to jointly reduction in both modes ($\mathbf{D}_{58\times70}$). In this case, both matrices $\mathbf{C}$ and $\mathbf{A}$ differ from their counterparts in the original dataset. Consequently, the basis vectors of both the concentration and spectral spaces are altered, leading to changes in both $\mathbf{U}$ and $\mathbf{V}$. According to the duality principle, these changes directly affect the outer polygons defined in the two spaces. Therefore, the outer polygons of the original and reduced datasets
differ in both the concentration and spectral spaces, as observed in the third row of Fig~\ref{fig3:procrustes}.


\begin{table*}
\centering
\caption{Correlation matrices corresponding to the $\mathbf{C}$ matrices of the (a) main and (b) reduced data}
\begin{tabular}{cc}
\begin{minipage}[t]{0.3\textwidth}
\subcaption{}
\[
\begin{pmatrix}
1 & 0.0527 & -0.5293\\
0.0527 & 1 & 0.0508\\
-0.5293 & 0.0508 & 1
\end{pmatrix}
\]
\end{minipage}
&
\begin{minipage}[t]{0.3\textwidth}
\subcaption{}
\[
\begin{pmatrix}
1 & 0.1793 & -0.4284\\
0.1793 & 1 & 0.1972\\
-0.4284 & 0.1972 & 1
\end{pmatrix}
\]
\end{minipage}
\end{tabular}
\label{tab:1.1}
\end{table*}

\subsubsection{Changes in the AFSs}

The AFSs are bounded by the inner and outer polygons. Therefore, any change in the shape, size of the inner and outer polygons directly affects the corresponding AFSs. Consequently, modifications in the geometry of the feasible regions resulting from data reduction are expected to be reflected in the shape and extent of the AFSs. In cases where the inner and outer polygons remain unchanged, the corresponding AFSs also preserve their geometry. In contrast, whenever changes occur in the inner and outer polygons, the corresponding AFSs exhibit noticeable differences in both shape and size. 

\subsection{Two-component data sets}
To investigate the effect of data reduction in two component data sets, a simulated two-component calibration dataset was generated. The corresponding spectral and concentration profiles are presented in Fig~\ref{fig6:2compdataset}. The original data matrix consisted of 100 rows and 200 columns. Data reduction was then performed under three different scenarios. In the first scenario, two essential rows were selected, resulting in a reduced data matrix of size $2 \times 200$. In the second scenario, two essential columns were selected, yielding a reduced matrix of size $100 \times 2$. In the third scenario, data reduction was simultaneously applied in both modes, producing a final reduced matrix of size $2 \times 2$. An interesting observation is that, for all two-component datasets, the data reduction procedure always results in a $2 \times 2$ reduced matrix. In contrast, for datasets containing more than two components, no general rule exists for the dimensions of the reduced matrix, as they depend on the geometric structure of the dataset and the number of essential points in both the row and column modes.
Following each data reduction strategy, the angles between the inner boundaries, the angles between the outer boundaries, and the angles between the inner and outer boundaries (AFSs) for each component were calculated. The obtained results for all cases are summarized in Table~\ref{tab:2.1} and \ref{tab:2.2}. In the following section, the variations in the inner and outer boundaries, as well as the changes in the AFSs, will be systematically analyzed and discussed. 


\begin{figure}
  \centering
  \includegraphics[width=0.45\textwidth]{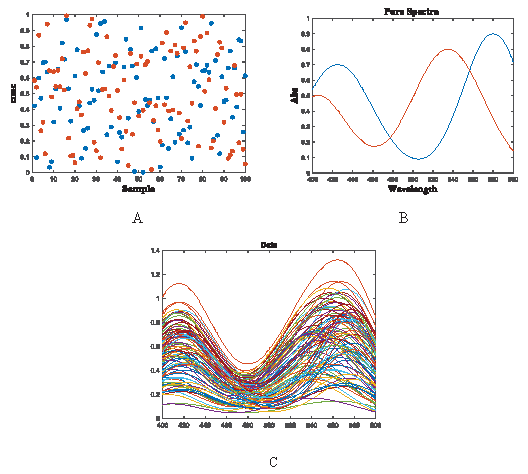} \\
  \caption{Two-component dataset. A) concentration of components in the samples, B) pure spectral profiles and C) simulated noise-free data}
  \label{fig6:2compdataset}
\end{figure}

\begin{table*}
  \centering
  \caption{Estimated angles between the inner boundaries, the angles between the outer boundaries, and the angles between the inner and outer boundaries for each component in spectral space}
  \label{tab:2.1}
  \begin{tabular}{ccccc}
  \toprule
  Data matrix &
  \makecell{Angle between\\inners} &
  \makecell{Angle between\\outers} &
  \makecell{Angle between outer and \\inner of 1st component}&
  \makecell{Angle between outer and \\inner of 2nd component}\\
  \midrule
  $100\times200$ & 41.01 & 57.96 & 10.31 & 6.65 \\
  $2\times200$ & 41.01 & 57.96 & 10.31 & 6.65 \\
  $100\times2$ & 68.13 & 90.00 & 13.11 & 8.76 \\
  $2\times2$ & 68.13 & 90.00 & 13.11 & 8.76 \\
  \bottomrule
  \end{tabular}
  \end{table*}

\begin{table*}
  \centering
  \caption{Estimated angles between the inner boundaries, the angles between the outer boundaries, and the angles between the inner and outer boundaries for each component in concentration space}
  \label{tab:2.2}
  \begin{tabular}{ccccc}
  \toprule
  Data matrix &
  \makecell{Angle between\\inners} &
  \makecell{Angle between\\outers} &
  \makecell{Angle between outer and \\inner of 1st component}&
  \makecell{Angle between outer and \\inner of 2nd component}\\
  \midrule
  $100\times200$ & 27.57 & 39.62 & 5.72 & 6.32 \\
  $2\times200$ & 64.83 & 90.00 & 19.32 & 5.85 \\
  $100\times2$ & 27.57 & 39.62 & 5.72 & 6.32 \\
  $2\times2$ & 64.83 & 90.00 & 19.32 & 5.85 \\
  \bottomrule
  \end{tabular}
  \end{table*}

\subsubsection{Changes in the inner boundaries}
 The objective of this analysis is to investigate the changes introduced by data reduction relative to the original dataset. By comparing the angles between the inner boundaries obtained from the original data with those obtained from the reduced datasets, it can be observed that certain data reduction strategies alter the angle between the inner boundaries. When the two essential rows are selected, the resulting reduced data matrix has dimensions of $2 \times 200$. In this case, the angle between the inner boundaries remains identical to that obtained from the original dataset in spectral space (Table~\ref{tab:2.1}, first two rows). This observation can be explained by examining the variance--covariance matrix in the row space. As shown in Table~\ref{tab:2.3and4}(a) and (b), the covariance submatrix corresponding to the two essential rows is identical in both the original dataset and the reduced dataset. Consequently, the geometric relationship between these two rows is preserved after data reduction.  Since both the essential rows and their corresponding variance--covariance structure remain unchanged, the angle between the inner boundaries is also preserved. Therefore, reducing the data matrix from $100 \times 200$ to $2 \times 200$ does not affect the inner-boundaries angle associated with these rows, leading to the same value observed for the original dataset. However, by comparing the first two rows of Table~\ref{tab:2.2}, it can be observed that when the data matrix is reduced to a size of $2 \times 200$, the angle between the inner boundaries in the concentration space changes compared to the original dataset. 

\begin{table}
\centering
\caption{Variance--covariance matrix corresponding essential point for (a) reduced and (b) original data in row space.}
\begin{tabular}{cc}
\begin{minipage}[t]{0.15\textwidth}
\subcaption{}
\[
\begin{pmatrix}
67.09 & 57.56\\
57.56 & 49.62
\end{pmatrix}
\]
\end{minipage}
&
\begin{minipage}[t]{0.15\textwidth}
\subcaption{}
\[
\begin{pmatrix}
67.09 & 57.56\\
57.56 & 49.62
\end{pmatrix}
\]
\end{minipage}
\end{tabular}
\label{tab:2.3and4}
\end{table}

The variance--covariance matrices of both the original and the reduced datasets in the column space are $200 \times 200$ matrices. However, these matrices are not related by a simple scalar scaling, and their elements are altered in a non-uniform (asymmetric) manner. As a consequence, the principal directions in the concentration space are modified, leading to a change in the orientation of the inner boundaries. 
When the data matrix is reduced to a size of $100 \times 2$, the opposite behavior is observed. In this case, the variance--covariance matrix corresponding to the essential points in the column space remains identical to that of the original dataset. Consequently, the angle between the inner boundaries in the concentration space is preserved, and no change is observed compared to the original data. In contrast, the variance--covariance matrix in the row space ($100 \times 100$) is altered due to the reduction process. As a result, the covariance structures of the original and reduced datasets in the row space are no longer equivalent. This leads to a change in the angle between the inner boundaries in the row space.
Finally, when data reduction is applied jointly in both modes, both row-space and column-space variance--covariance matrices differ from those of the original dataset. Consequently, the angles between the inner boundaries in both spaces change compared to the original data.

\subsubsection{Changes in the outer boundaries}
When the data matrix is reduced to a size of $2 \times 200$, the basis vectors of the row space remain exactly the same as those of the original dataset. Therefore, the outer boundaries defined in this space are preserved, and no change is observed between the original and reduced datasets. This is consistent with Table~\ref{tab:2.1}1, where the angle between the outer boundaries in the original data and the $2 \times 200$ reduced data remains identical.
In contrast, in the column space, the situation is different. In the original dataset, the basis vectors of the column space consist of two vectors of dimension $1 \times 200$, with an angle of approximately 38.6016° between them. However, in the reduced dataset ($2 \times 200$), the corresponding basis vectors become two vectors of dimension $1 \times 2$, and the angle between them changes to approximately 87.0748°. Since the basis vectors of the column space are different in the original and reduced datasets, both in terms of dimensionality and angular relationship, the outer boundaries defined by these bases are also altered. As a result, the angle between the outer boundaries in the column space changes when comparing the reduced dataset with the original one.
When the data matrix is reduced to a size of $100 \times 2$, the opposite behavior is observed. In this case, the basis vectors of the column space remain identical in both the original and reduced datasets. Consequently, the outer boundaries defined in this space are also preserved, and no change is observed in the angle between the outer boundaries when comparing the original and reduced data.  In contrast, the situation in the row space is different. In the original dataset, the basis vectors of the row space consist of two vectors of dimension $1 \times 200$ with an angle of approximately 41.9793°. However, in the reduced dataset ($100 \times 2$), the corresponding row-space basis vectors become two vectors of dimension $1 \times 2$, while the angular relationship between them remains approximately 69.4912°. Since the basis vectors of the row space are different in the original and reduced datasets, both in terms of dimensionality and angular relationship, the outer boundaries defined by these bases are also altered. As a result, the angle between the outer boundaries in the row space changes when comparing the reduced dataset with the original one.
Finally, when data reduction is applied jointly in both modes, the basis vectors and their angular relationships change in both the row and column spaces. Consequently, the outer boundaries defined in both spaces are modified, leading to changes in the angles between the outer boundaries in both representations of the reduced data compared to the original dataset.

\subsubsection{Changes in the AFSs}
AFSs are constrained by the inner and outer boundaries. Therefore, in cases where the inner and outer boundaries are altered in the reduced dataset compared to the original data, the rotational ambiguity is also affected. In other words, the angle between the inner and outer boundaries for each component changes as a result of data reduction. In two-component data, the reduced data matrix is always of size 2×2, meaning that the profiles of the reduced data consist of only two elements. Since profiles corresponding to the outer boundary represent the non-negativity constraint, and any profile on this boundary must contain at least one zero element, it follows that in the reduced two-component case, one of the two elements is always zero. Consequently, the angle between the two outer components in the reduced space is always 90 degrees. Importantly, this property holds for all two-component data sets. 
The off-diagonal elements of the variance--covariance matrix correspond to the inner products between pairs of rows or columns of the data matrix. Since the inner product between two vectors can be expressed as the product of their norms and the cosine of the angle between them, any change in these off-diagonal elements directly reflects a change in the angles between the corresponding vectors. Consequently, when data reduction leads to changes in the variance--covariance structure associated with essential points, it implies that the geometric relationships between essential rows and columns are altered. This ultimately results in a change in the data pattern, which in turn leads to variations in rotational ambiguity.

\subsection{Comparison of the inner polygon of the original and reduced datasets after external normalization}

A three-component simulated dataset was generated to investigate whether the 
inner polygon preserves its geometric structure following data reduction using external normalization~\cite{RAJKO2024}. The simulated spectral profiles are shown in Fig~\ref{fig7:extnormdata}. 

\begin{figure}
  \centering
  \includegraphics[width=0.45\textwidth]{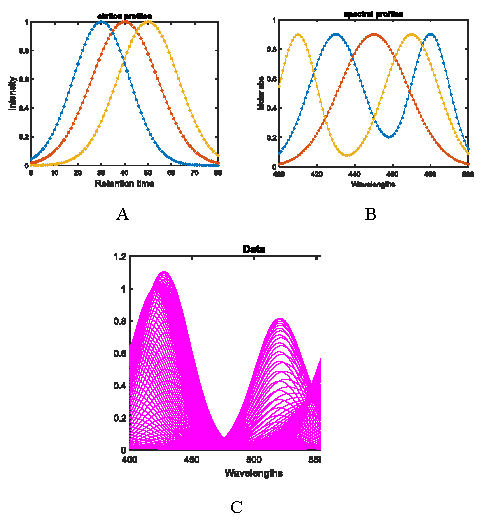} \\
  \caption{A) Elution profiles, B) Pure spectral profiles and C) simulated noise free data}
  \label{fig7:extnormdata}
\end{figure}

Step 1: Selection of Essential Rows and Columns. 
To identify the essential rows, each row of the data matrix was first normalized 
by dividing it by the sum of its elements. SVD 
was then applied to the normalized data to obtain the row coordinates in the SVD 
space. The resulting coordinates were subsequently mean-centered, after which a 
second SVD was performed to obtain the final row coordinates. Finally, the 
vertices of the convex hull in the row space were identified using MATLAB's 
convhull function, and these vertices were selected as the essential rows. 
To identify the essential columns, each column of the data matrix was first 
normalized by dividing it by the sum of its elements. SVD was then applied to the normalized data to obtain the 
column coordinates in the SVD space. The resulting coordinates were 
subsequently mean-centered, and a second SVD was performed to obtain the final 
column coordinates. The vertices of the convex hull in the column space were 
then identified using MATLAB's convhull function, and these vertices were 
selected as the essential columns. 
Using this procedure, a total of 64 essential rows and 36 essential columns were 
identified. Finally, a reduced data matrix was constructed by retaining only the 
identified essential rows and essential columns from the original data matrix. 
Comparison of the Inner Polygons in the Row Space. 
The original and reduced data matrices were compared by evaluating their inner 
polygons in the row score space. For both datasets, each row was first normalized 
by dividing it by the sum of its elements. SVD 
was then applied to the normalized data to obtain the row scores. The resulting 
row scores were subsequently mean-centered, and a second SVD was performed 
to obtain the final row scores. The corresponding inner polygons were then 
constructed in the row score space. 
To enable a meaningful comparison, the inner polygon of the reduced data matrix 
was aligned with that of the original data matrix using Procrustes transformation. 
Specifically, the inner polygon of the original data matrix was considered as the 
reference, and the transformed inner polygon of the reduced data matrix was 
superimposed onto it. As shown in Figure 8A, the two inner polygons do not 
completely coincide, indicating that the reduced data matrix does not fully 
preserve the geometric structure of the original data matrix in the row score space. 
Comparison of the Inner Polygons in the Column Space. 
A similar analysis was performed in the column score space. For both the original 
and the reduced data matrices, each column was first normalized by dividing it 
by the sum of its elements. SVD was then applied separately to the normalized data to obtain the column scores. The resulting column scores were subsequently mean-centered, and a second SVD was performed to obtain the final column scores. The corresponding inner polygons were then constructed in the column score space. 
Before comparison, the inner polygon of the reduced data matrix was aligned 
with that of the original data matrix using Procrustes transformation, with the inner polygon of the original data matrix serving as the reference. As shown in Fig~\ref{fig8:externaln}B, the two inner polygons do not completely coincide, indicating that the reduced data matrix does not fully preserve the geometric structure of the original data matrix in the column score space. 

\begin{figure}
  \centering
  \includegraphics[width=0.45\textwidth]{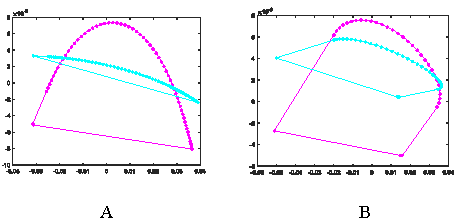} \\
  \caption{A) row space and B) column space. In both figures, the pink inner 
polygon represents the original data matrix, while the cyan inner polygon 
represents the reduced data matrix after Procrustes transformation.}
  \label{fig8:externaln}
\end{figure}

\subsection{Real data sets}
This dataset comprises real HPLC-DAD measurements of organophosphorus pesticides in natural water samples, originally collected as part of an interlaboratory comparison exercise. Dataset consists of a three-component chemical system containing two identified pesticides, diazinon and ethyl parathion, along with one unknown interferent (Fig~\ref{fig:real_data}). Further details regarding the experimental procedure, data acquisition, and dataset are provided in~\cite{AZZOUZ2008}. 
The dataset was reduced in three different ways to investigate the effect of dimensional reduction on the Borgen-Rajk\'{o} plot representation. In the first case, the essential rows were retained while the non-essential rows were removed, resulting in the reduced data matrix $\mathbf{D}_{dr}$. In the second case, the essential columns were retained, and the non-essential columns were removed, yielding $\mathbf{D}_{dc}$. Finally, in the third case, both non-essential rows and columns were removed, resulting in the jointly reduced data matrix $\mathbf{D}_{drc}$.
The Borgen-Rajk\'{o} plots were then calculated for the original dataset and for each of the three reduced datasets. For each case, the Borgen-Rajk\'{o} plots were evaluated in both the row and column spaces to assess the changes in the corresponding geometric representations resulting from dimensional reduction. The resulting plots are presented in Fig~\ref{fig:realBRplt}.

\begin{figure}
  \centering
  \includegraphics[width=0.3\textwidth]{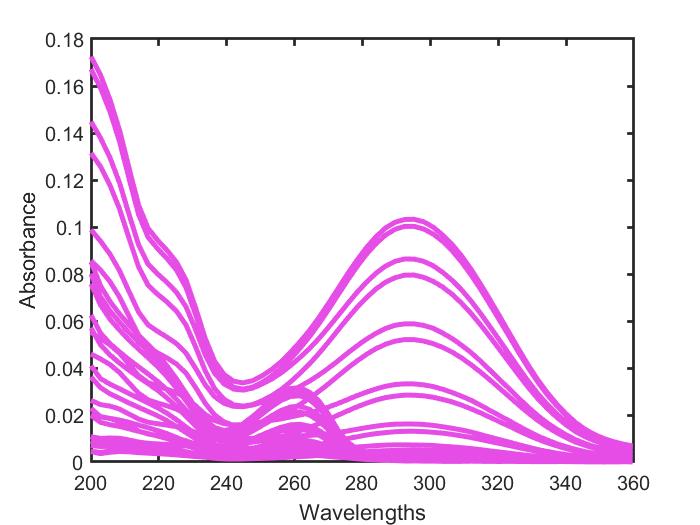} \\
  \caption{Real HPLC-DAD measurements of organophosphorus pesticides.}
  \label{fig:real_data}
\end{figure}

\begin{figure}
  \centering
  \includegraphics[width=0.45\textwidth]{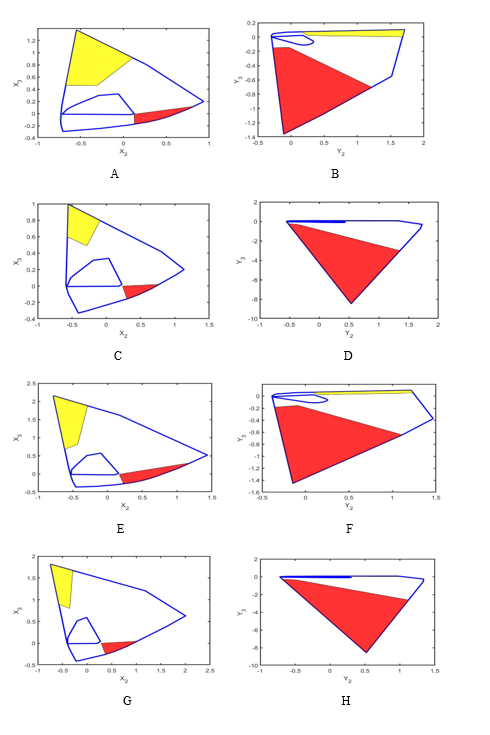} \\
  \caption{Borgen-Rajk\'{o} plots obtained for the original and reduced datasets in the row and column spaces: (A) original dataset in the row space; (B) original dataset in the column space; (C) $\mathbf{D}_{dr}$ in the row space; (D) $\mathbf{D}_{dr}$ in the column space; (E) $\mathbf{D}_{dc}$ in the row space; (F) $\mathbf{D}_{dc}$ in the column space; (G) $\mathbf{D}_{drc}$ in the row space; and (H) $\mathbf{D}_{drc}$ in the column space.}
  \label{fig:realBRplt}
\end{figure}

The results obtained from the real dataset are consistent with those observed for the simulated datasets. Following dimensional reduction, the pattern of the data points changes, which consequently alters the geometry of the inner polygon. Thus, the inner polygon obtained from a reduced dataset may undergo a transformation relative to that of the original dataset. Depending on the type of dimensional reduction, this transformation may correspond to either an orthogonal transformation or a non-orthogonal isomorphism. Therefore, the Borgen-Rajk\'{o} plots of the original and reduced datasets cannot be directly compared based solely on their apparent geometric orientation. A correction step is consequently required to transform the reduced-space representation into the corresponding reference orientation before a meaningful comparison with the original Borgen-Rajk\'{o} plot can be made. 

\section{Theoretical considerations}
\subsection{The four fundamental subspaces}
A matrix can be thought of as a machine that transforms input vectors into output vectors. The "four fundamental subspaces" describe what that machine does and does not do — where its outputs can land, and which inputs get erased.
Thus, any linear system — whether a physical sensor, a statistical model, or a network of equations — has four intrinsic "geometric compartments" that describe what information flows through it, what gets lost, and what can or cannot be produced. These are not artifacts of notation; they are structural properties of the system itself.

\subsubsection{Column space — the reachable outputs} 
Every linear system has a fixed repertoire of outcomes it can produce, no matter what input it receives. This is the column space: the full range of physically achievable outputs. A measurement device, for instance, can only report signals that lie within its column space — any true phenomenon outside that range is invisible to it, not because the phenomenon doesn't exist, but because the instrument has no channel through which to express it.

\subsubsection{Row space — the information that survives measurement}
Of everything that could vary in the input, only some directions actually influence the output; the rest are indistinguishable from the system's point of view. The row space is the portion of the input's structure that the system is actually sensitive to. In an experimental setting, this is the subset of the input variability that leaves a detectable trace in the data — the effective "resolution" of the system.

\subsubsection{Null space — the information destroyed by the process}
The null space consists of every input pattern that produces zero effect — patterns the system is fundamentally blind to. This is the scientific notion of an unobservable degree of freedom: a variation in the underlying cause that, however large, leaves absolutely no signature in the outcome. It defines the intrinsic limit of what can ever be inferred by inverting the process.

\subsubsection{Left null space — the inaccessible region of outcome-space}
Complementary to the column space, the left null space identifies combinations of output measurements that no input, however chosen, could ever generate. Practically, this is where noise, error, or model misspecification is forced to live: any discrepancy between an observed outcome and the system's true capability necessarily falls into this compartment, since the system itself has no mechanism to produce that discrepancy from a genuine input.

\subsubsection*{The unifying picture}
Together, these four regions partition reality into what a system can convey and what it inevitably discards, on both the input and output side:
\begin{itemize}
\item On the input side: some structure gets through (row space); the rest is annihilated (null space).
\item On the output side: some region is achievable (column space); the rest is permanently out of reach (left null space).
\end{itemize}
This decomposition is why the same mathematical structure appears across disciplines — in signal processing it explains bandwidth limits and irrecoverable noise; in statistics it explains identifiability and residual error; in physics it explains conserved versus dissipated quantities. In every case, the "four subspaces" describe the same universal fact: any linear process has a built-in boundary between what it can transmit and what it necessarily loses, on both ends of the transformation. More technical details can be found in Section~\ref{App:Sec-4-fundSubSp} in Appendix.

\subsection{The four fundamental subspaces under convex-hull-vertex deletion}
This and the next subsections were compiled based on the relevant information published in the literature~\cite{ROCKAFELLAR1970,Cheng2005,Mahoney2009,Ding2018,Cortinovis2020}.

Given $\mathbf{D}\in\mathbb{R}^{m\times n}$, rank $r$, with compact SVD $\mathbf{D}=\mathbf{U}\, \mathbf{S}\,\mathbf{V}^\mathrm{T}$, $\mathbf{X}=\mathbf{U}\,\mathbf{S}$, $\mathbf{Y}=\mathbf{V}\,\mathbf{S}$. Let $\mathcal{I} \in \{i_1,i_2,\ldots,i_p\} \subseteq \{1,\dots,m\}$ index kept (hull-vertex) rows and $\mathcal{J} \in \{j_1,j_2,\ldots,j_q\} \subseteq \{1,\dots,n\}$ index kept (hull-vertex) columns, chosen so rank is preserved throughout.

\begin{strip}
\begin{samepage}
\textbf{Column space}
 $\mathfrak{Col}(\mathbf{D})\subset\mathbb{R}^m$
\begin{center}
\begin{tabular}{@{}>{\raggedright}p{3.2cm} p{10cm}@{}}
\toprule
\textbf{Matrix} & \textbf{What happens} \\
\midrule
$\mathbf{D}$ & $\mathfrak{Col}(\mathbf{D}) = \operatorname{span}(\mathbf{U})$, dimension $r$ \\[4pt]
$\mathbf{D}_{dr}$ (rows deleted) & Ambient space shrinks $\mathbb{R}^m \to \mathbb{R}^{p}$; $\mathbf{U}(\mathcal{I},:)$ relates to $\mathbf{U}_{dr}$ by a general invertible (non-orthogonal) map $N_{dr}$ --- \emph{isomorphic}, not identical \\[4pt]
$\mathbf{D}_{dc}$ (columns deleted) & $m$ unchanged; $\mathfrak{Col}(\mathbf{D}_{dc}) = \mathfrak{Col}(\mathbf{D})$ \emph{exactly} --- kept columns already span it. $\mathbf{U}_{dc} = \mathbf{U}\, \mathbf{Q}_{dc}$, $\mathbf{Q}_{dc}$ orthogonal \\[4pt]
$\mathbf{D}_{drc}$ & Combination of both effects above \\
\bottomrule
\end{tabular}
\end{center}
\end{samepage}
\end{strip}

\begin{strip}
\begin{samepage}
\textbf{Row space} $\mathfrak{Col}(\mathbf{D}^\mathsf{T}) \subset \mathbb{R}^n$
\begin{center}
\begin{tabular}{@{}>{\raggedright}p{3.2cm} p{10cm}@{}}
\toprule
\textbf{Matrix} & \textbf{What happens} \\
\midrule
$\mathbf{D}$ & $\mathfrak{Col}(\mathbf{D}^\mathsf{T}) = \operatorname{span}(\mathbf{V})$, dimension $r$ \\[4pt]
$\mathbf{D}_{dr}$ & $n$ unchanged; $\mathfrak{Col}(\mathbf{D}_{dr}^\mathsf{T}) = \mathfrak{Col}(\mathbf{D}^\mathsf{T})$ \emph{exactly} --- deleted rows are combinations of kept ones. $\mathbf{V}_{dr} = \mathbf{V}\, \mathbf{Q}_{dr}$, $\mathbf{Q}_{dr}$ orthogonal \\[4pt]
$\mathbf{D}_{dc}$ & Ambient shrinks $\mathbb{R}^n \to \mathbb{R}^{q}$; $\mathbf{V}(\mathcal{J},:)$ relates to $\mathbf{V}_{dc}$ by a general invertible map $M$ --- \emph{isomorphic}, not identical \\[4pt]
$\mathbf{D}_{drc}$ & $\mathbf{V}_{drc} = \mathbf{V}_{dc}\, \mathbf{Q}_{drc}$, a further rotation on top of the isomorphism already applied for $\mathbf{V}_{dc}$ \\
\bottomrule
\end{tabular}
\end{center}
\end{samepage}
\end{strip}

\textbf{Null space} $\mathfrak{Null}(\mathbf{D}) \subset \mathbb{R}^n$
\begin{itemize}
  \item Dimension $n - r$ for $\mathbf{D}$; shrinks to $q - r$ for $\mathbf{D}_{dc}$/$\mathbf{D}_{drc}$ since fewer columns remain.
  \item Row deletion alone ($\mathbf{D}_{dr}$) leaves the ambient space $\mathbb{R}^n$ untouched, so $\mathfrak{Null}(\mathbf{D}_{dr})$ still has dimension $n - r$, and in fact $\mathfrak{Null}(\mathbf{D}_{dr})$ can differ as a \emph{set} even though the dimension matches --- because $\mathfrak{Null}(\mathbf{D}_{dr})$ only requires the $\mathbf{X}(\mathcal{I},:)$-side relations to vanish, a weaker condition than requiring all of $\mathbf{D}$'s rows to vanish. In practice
  $
  \mathfrak{Null}(\mathbf{D}_{dr}) \supseteq \mathfrak{Null}(\mathbf{D}).
  $
  \item If the kept columns $\mathcal{J}$ form an exact basis ($|\mathcal{J}| = r$), the null space collapses to $\{0\}$.
\end{itemize}

\textbf{Left null space} $\mathfrak{Null}(\mathbf{D}^\mathsf{T}) \subset \mathbb{R}^m$
\begin{itemize}
  \item Dimension $m - r$ for $\mathbf{D}$; shrinks to $p - r$ for $\mathbf{D}_{dr}$/$\mathbf{D}_{drc}$.
  \item Column deletion alone ($\mathbf{D}_{dc}$) leaves $\mathbb{R}^m$ untouched, so $\mathfrak{Null}(\mathbf{D}_{dc}^\mathsf{T})$ has dimension $m - r$ and
  $
  \mathfrak{Null}(\mathbf{D}_{dc}^\mathsf{T}) \supseteq \mathfrak{Null}(\mathbf{D}^\mathsf{T}).
  $
  \item If the kept rows $\mathcal{I}$ form an exact basis ($|\mathcal{I}| = r$), the left null space collapses to $\{0\}$.
\end{itemize}

The subspace living in the ambient space that is \textbf{not} being shrunk is preserved \emph{exactly} (only rotates, via an orthogonal $\mathbf{Q}$); the subspace living in the ambient space that \textbf{is} being shrunk survives only as an \emph{isomorphic image} under the coordinate-restriction map (via a general invertible, non-orthogonal transform). This is the structural fact verified numerically for $\mathbf{U}, \mathbf{V}, \mathbf{X}, \mathbf{Y}$ across $\mathbf{D}, \mathbf{D}_{dr}, \mathbf{D}_{dc}, \mathbf{D}_{drc}$.

\subsection{Subspace equivalence under convex-hull-vertex deletion}\label{subsec:subspEqunderCHdel}

\subsubsection*{Setup}

Let $\mathbf{D} \in \mathbb{R}^{m\times n}$ have rank $r$, with compact SVD: 
$
\mathbf{D} = \mathbf{U}\, \mathbf{S}\, \mathbf{V}^\mathsf{T}, \quad \mathbf{U}^\mathsf{T}\,\mathbf{U} = \mathbf{I}_r,\ \mathbf{V}^\mathsf{T}\,\mathbf{V} = \mathbf{I}_r,\ \mathbf{S} = \operatorname{diag}(\sigma_1,\dots,\sigma_r),
$
and write $\mathbf{X} := \mathbf{U}\,\mathbf{S}$, so that $\mathbf{D} = \mathbf{X}\,\mathbf{V}^\mathsf{T}$.

Let $\mathcal{I} \in \{i_1,i_2,\ldots,i_p\} \subseteq \{1,\dots,m\}$ index the row-hull vertices and $\mathcal{J} \in \{j_1,j_2,\ldots,j_q\} \subseteq \{1,\dots,n\}$  the column-hull vertices, chosen so that deletion is rank-preserving:
\[
\operatorname{rank}\bigl(\mathbf{D}(\mathcal{I},:)\bigr) = \operatorname{rank}\bigl(\mathbf{D}(:,\mathcal{J})\bigr) = \operatorname{rank}\bigl(\mathbf{D}(\mathcal{I},\mathcal{J})\bigr) = r.
\]
Define the reduced factorizations $\mathbf{D}_{dr}=\mathbf{X}_{dr}\,\mathbf{V}_{dr}^\mathsf{T}$, $\mathbf{D}_{dc}=\mathbf{X}_{dc}\,\mathbf{V}_{dc}^\mathsf{T}$, $\mathbf{D}_{drc}=\mathbf{X}_{drc}\,\mathbf{V}_{drc}^\mathsf{T}$, each obtained from the compact SVD of the corresponding submatrix.

\subsubsection*{A necessary caution}

If we additionally insist that \emph{every} $\mathbf{X}$ has exactly orthonormal columns (with no residual diagonal factor), then $\mathbf{D}, \mathbf{D}_{dr}, \mathbf{D}_{dc}, \mathbf{D}_{drc}$ are all forced to be \textbf{partial isometries} (rank-$r$ matrices whose nonzero singular values equal $1$), since
\[
\mathbf{D}^\mathsf{T}\,\mathbf{D} = \mathbf{V}\, \mathbf{X}^\mathsf{T}\mathbf{X}\, \mathbf{V}^\mathsf{T} = \mathbf{V}\,\mathbf{V}^\mathsf{T}
\]
would then be a projection with eigenvalues $0,1$. But deleting rows or columns generically \emph{shrinks} the singular values via the rank-one downdate $\tilde{\boldsymbol{\Sigma}}^2 = \boldsymbol{\Sigma}^2 - \mathbf{z}\,\mathbf{z}^\mathsf{T}$. Hence requiring $\mathbf{X}_{dr}, \mathbf{X}_{dc}$, etc.\ to remain exactly orthonormal is consistent only in a non-generic special case; we identify this case precisely below, and give the fully general (and always valid) statement afterward.

\subsubsection*{Step 1: Facts that hold unconditionally}

\begin{lemma}[Row deletion preserves the row space exactly]
	If every deleted row is a linear combination of the kept rows $\mathcal{I}$ (true in particular when it is a convex combination), then
	\[
	\mathfrak{Row}\bigl(\mathbf{D}(\mathcal{I},:)\bigr) = \mathfrak{Row}(\mathbf{D}) =: S_{\mathrm{row}}.
	\]
\end{lemma}
\begin{proof}
	For $i \notin \mathcal{I}$, $d_i = \sum_{j\in I} c_j d_j \in \mathfrak{Row}(\mathbf{D}(\mathcal{I},:))$, so every row of $\mathbf{D}$ lies in $\mathfrak{Row}(\mathbf{D}(\mathcal{I},:))$, giving $\mathfrak{Row}(\mathbf{D}) \subseteq \mathfrak{Row}(\mathbf{D}(\mathcal{I},:))$. The reverse inclusion is trivial.
\end{proof}

Since $\mathbf{V}$ and $\mathbf{V}_{dr}$ are both orthonormal bases of the same subspace $S_{\mathrm{row}}$, any two orthonormal bases of the same subspace differ by an orthogonal matrix:
\[
\boxed{\mathbf{V}_{dr} = \mathbf{V}\, \mathbf{Q}_{dr}}, \qquad \mathbf{Q}_{dr} \in \mathfrak{Orth}(r).
\]

\begin{lemma}[Column deletion preserves the column space exactly]
	Dually, if $\mathcal{J}$ indexes kept columns spanning $\mathfrak{Col}(\mathbf{D})$, then $\mathfrak{Col}(\mathbf{D}(:,\mathcal{J})) = \mathfrak{Col}(\mathbf{D})$, and $\mathbf{U}, \mathbf{U}_{dc}$ (both orthonormal bases of $\mathfrak{Col}(\mathbf{D})$) satisfy
	\[
	\mathbf{U}_{dc} = \mathbf{U}\, \mathbf{Q}_{dc}, \qquad \mathbf{Q}_{dc} \in \mathfrak{Orth}(r).
	\]
\end{lemma}

\subsubsection*{Step 2: What the tension reveals}

Matching $\mathbf{D}(\mathcal{I},:) = \mathbf{X}(\mathcal{I},:)V^\mathsf{T}$ with $\mathbf{D}_{dr} = \mathbf{X}_{dr}\,\mathbf{V}_{dr}^\mathsf{T} = \mathbf{X}_{dr}\,\mathbf{Q}_{dr}^\mathsf{T}\mathbf{V}^\mathsf{T}$ and cancelling $\mathbf{V}^\mathsf{T}$ (right-multiply by $\mathbf{V}$, using $\mathbf{V}^\mathsf{T}\,\mathbf{V} = \mathbf{I}_r$):
\[
\boxed{\mathbf{X}_{dr} = \mathbf{X}(\mathcal{I},:)\,\mathbf{Q}_{dr}}.
\]
This relation is \emph{forced}, not chosen. Now
\[
\mathbf{X}_{dr}^\mathsf{T}\,\mathbf{X}_{dr} = \mathbf{Q}_{dr}^\mathsf{T}\,\mathbf{X}(\mathcal{I},:)^\mathsf{T}\,\mathbf{X}(\mathcal{I},:)\,\mathbf{Q}_{dr},
\]
which equals $\mathbf{I}_r$ if and only if $\mathbf{X}(\mathcal{I},:)^\mathsf{T}\,\mathbf{X}(\mathcal{I},:) = \mathbf{I}_r$, i.e.\ if and only if
\[
\mathbf{X}(\mathcal{I}^c,:)^\mathsf{T}\,\mathbf{X}(\mathcal{I}^c,:) = \mathbf{X}^\mathsf{T}\,\mathbf{X} - \mathbf{X}(\mathcal{I},:)^\mathsf{T}\,\mathbf{X}(\mathcal{I},:) = 0,
\]
i.e.\ \textbf{the deleted rows of $\mathbf{X}$ are exactly zero.} This means the deleted data points $\mathbf{D}(i,:) = \mathbf{X}(i,:)\,\mathbf{V}^\mathsf{T}$ were literally the zero vector — a very special case, not the generic ``interior point of the convex hull.''

By the symmetric argument applied to column deletion, $\mathbf{V}_{dc}$ is exactly orthonormal only if the deleted columns of $\mathbf{V}$, i.e.\ $\mathbf{V}(\mathcal{J}^c,:)$, vanish exactly.

\begin{claim}
	Free rotations exist on the side where the ambient space does not shrink (Step 1). On the side where the ambient space does shrink ($\mathbf{X}$ under row deletion, $\mathbf{V}$ under column deletion), exact orthonormality of the reduced factor is only consistent with an extra vanishing condition; otherwise a singular-value correction is unavoidable.
\end{claim}

\subsubsection*{Step 3: The two honest regimes}

\subsubsection*{Regime 1 (idealized): deleted mass is orthogonal-complement-negligible}

If the discarded rows/columns satisfy $\mathbf{X}(\mathcal{I}^c,:) = 0$ and $\mathbf{V}(\mathcal{J}^c,:) = 0$ exactly, every reduced factor stays exactly orthonormal, and:

\begin{center}
	\begin{tabular}{@{}ll@{}}
		\toprule
		Pair & Exact relation \\
		\midrule
		$\mathbf{V} \leftrightarrow \mathbf{V}_{dr}$ & $\mathbf{V}_{dr} = \mathbf{V}\,\mathbf{Q}_{dr}$ \quad (rotation in $\mathbb{R}^n$) \\
		$\mathbf{X} \leftrightarrow \mathbf{X}_{dc}$ & $\mathbf{X}_{dc} = \mathbf{X}\, \mathbf{Q}_{dc}$ \quad (rotation in $\mathbb{R}^m$) \\
		$\mathbf{X} \leftrightarrow \mathbf{X}_{dr}$ & $\mathbf{X}_{dr} = \mathbf{X}(\mathcal{I},:)\,\mathbf{Q}_{dr}$ \quad (isometric embedding) \\
		$\mathbf{V} \leftrightarrow \mathbf{V}_{dc}$ & $\mathbf{V}_{dc} = \mathbf{V}(\mathcal{J},:)\,\mathbf{Q}_{dc}$ \quad (isometric embedding) \\
		\bottomrule
	\end{tabular}
\end{center}

with $\mathbf{V}_{drc} = \mathbf{V}_{dc}\, \mathbf{Q}_{drc}$ and $\mathbf{X}_{drc} = \mathbf{X}_{dc}(\mathcal{I},:)\, \mathbf{Q}_{drc}$ by the same commuting-square argument. In this regime, all four subspaces are literally isometric copies of one another, and the transformations reduce to pure rotations $\mathbf{Q}_{dr}, \mathbf{Q}_{dc}, \mathbf{Q}_{drc} \in \mathfrak{Orth}(r)$.

\subsubsection*{Regime 2 (generic): reinstate the singular values}

For a generic convex-hull deletion, the vanishing condition fails, and the correct compact SVD is
$
\mathbf{D} = \mathbf{X}\, \boldsymbol{\Sigma}\, \mathbf{V}^\mathsf{T}, \, \mathbf{X}^\mathsf{T}\, \mathbf{X} = \mathbf{V}^\mathsf{T}\, \mathbf{V} = \mathbf{I}_r,\, \Sigma = \operatorname{diag}(\sigma_1,\dots,\sigma_r).
$
The subspace relations still hold exactly:
$
\mathbf{V}_{dr} = \mathbf{V}\, \mathbf{Q}_{dr}, \, \mathbf{U}_{dc} = \mathbf{U}\, \mathbf{Q}_{dc},
$
while the ambient-shrinking maps are general invertible (not orthogonal) matrices:
$
\mathbf{X}(\mathcal{I},:) = \mathbf{X}_{dr}\, \mathbf{N}_{dr}, \, \mathbf{V}(\mathcal{J},:) = \mathbf{V}_{dc}\, \mathbf{M},
$
and the singular values re-diagonalize through a rank-one downdate,
$
\boldsymbol{\Sigma}_{dr}^2 = \boldsymbol{\Sigma}^2 - \mathbf{z}\,\mathbf{z}^\mathsf{T} \quad(\text{restricted appropriately}),
$
so deletion is a genuine rank-one downdate of $\boldsymbol{\Sigma}^2$, not a pure rotation, whenever the ambient dimension of that factor shrinks.

\subsubsection*{Interim conclusions related to this subsection}

\begin{itemize}
	\item \textbf{Subspace equivalence} holds fully rigorously in both regimes: $\mathbf{V}, \mathbf{V}_{dr}$ are literally identical subspaces of $\mathbb{R}^n$; $\mathbf{U}, \mathbf{U}_{dc}$ (equivalently the column space) are literally identical subspaces of $\mathbb{R}^m$; $\mathbf{V}_{dc}, \mathbf{V}_{drc}$ and their row-restricted counterparts are isomorphic images under the coordinate-projection maps established above.
	\item The claim that \textbf{all} $\mathbf{X}$'s and $\mathbf{V}$'s stay exactly orthonormal with no diagonal correction is true only when the deleted rows/columns contribute nothing to the respective Gram matrices — a non-generic special case. In general, exact orthonormality of both factors simultaneously is incompatible with the singular-value shrinkage that necessarily occurs under deletion, and the fully general statement must reinstate the diagonal matrices $\boldsymbol{\Sigma}_{dr}, \boldsymbol{\Sigma}_{dc}, \boldsymbol{\Sigma}_{drc}$.
\end{itemize}

\section{Discussions}


It seems, that experimental investigations in Section~\ref{sec:exprmntl} could prove that the deletion of rows and/or columns of the original data matrix using the essential data points can cause discrepancies. 

\subsection{Illustration in a naive way}
We illustrate the subspace-equivalence theory on a real source-apportionment dataset from Ref.~\cite{HENRY1990}. The source-composition matrix $\mathbf{C} \in \mathbb{R}^{10\times 3}$ gives the mass fraction of each of ten chemical elements in three pollution sources (Marine, Urban Dust, Auto). The source-apportionment matrix $\mathbf{A} \in \mathbb{R}^{20\times 3}$ gives the estimated contribution ($\mu\text{g}/\text{m}^3$) of each source to twenty air samples. 

The data matrix analyzed is $\mathbf{D} = \mathbf{C}\,\mathbf{A}^\mathsf{T} \in \mathbb{R}^{10 \times 20}$ (elements $\times$ samples), of rank $r=3$, with compact SVD $\mathbf{D} = \mathbf{U}_m\,\mathbf{S}_m\,\mathbf{V}_m^\mathsf{T}$, giving
$
\mathbf{X} = \mathbf{U}_m\,\mathbf{S}_m,\qquad \mathbf{V} = \mathbf{V}_m, \quad \mathbf{Y} = \mathbf{V}_m\,\mathbf{S}_m, \quad \mathbf{U} = \mathbf{U}_m,
$
(after a standard sign-flip normalization to fix the loading signs).

The EDPs are the vertices of the convex hull of the row-points (elements, in $\mathbf{X}$-space) and of the column-points (samples, in $\mathbf{Y}$-space). For this dataset, the hull-defining index sets are
$
\text{EssI} = \{1, 4, 7, 9, 10\} \subset \{1,\dots,10\}
$, 
$
\text{EssJ} = \{3, 5, 6, 9, 15, 17, 18\} \subset \{1,\dots,20\}
$.
Deleting all rows outside $\text{EssI}$ (resp.\ columns outside $\text{EssJ}$) yields the rank-preserving reductions $\mathbf{D}_{dr} = \mathbf{C}(\text{EssI},:)\,\mathbf{A}^\mathsf{T}$, $\mathbf{D}_{dc} = \mathbf{C}\,\mathbf{A}(\text{EssJ},:)^\mathsf{T}$, and using both,  $\mathbf{D}_{drc} = \mathbf{C}(\text{EssI},:)\,\mathbf{A}(\text{EssJ},:)^\mathsf{T}$.

For each reduction we compute the change-of-basis matrix $\mathbf{T}$ solving the relevant rotation equation (e.g.\ $\mathbf{T} = \mathbf{V}_{dr}\backslash \mathbf{V}$), transform the full loadings through $\mathbf{T}$, and check that the transformed rows corresponding to the essential indices reproduce the loadings computed directly from the reduced matrix, up to $10$ decimal places.

\subsubsection*{$\mathbf{V}$ and $\mathbf{V}_{dr}$: row-space rotation}
With $\mathbf{T} = \mathbf{V}_{dr} \backslash \mathbf{V}$ and $\mathbf{X}_{dr}^{\mathrm{tr}} := \mathbf{X}/\mathbf{T}$, the rows of $\mathbf{X}_{dr}^{\mathrm{tr}}$ indexed by $\text{EssI}$ match $\mathbf{X}_{dr}$ exactly:

\begin{strip}
\begin{center}
	\begin{tabular}{@{}c ccc c ccc@{}}
		\toprule
		Row & \multicolumn{3}{c}{$\mathbf{X}_{dr}^{\mathrm{tr}}$} & Match & \multicolumn{3}{c}{$\mathbf{X}_{dr}$} \\
		\midrule
		1  & 9.3474 & $-2.1987$ & 0.36975   & 1 & 9.3474 & $-2.1987$ & 0.36975 \\
		2  & 10.873 & $-2.2936$ & 5.8425    & --- & --- & --- & --- \\
		3  & 26.274 & $-6.056$  & 14.755    & --- & --- & --- & --- \\
		4  & 9.6726 & $-1.4383$ & $-0.48233$ & 2 & 9.6726 & $-1.4383$ & $-0.48233$ \\
		5  & 1.5103 & $-0.33985$ & 0.66519  & --- & --- & --- & --- \\
		6  & 3.8363 & $-0.56715$ & 1.5889   & --- & --- & --- & --- \\
		7  & 0.74016 & $-0.17706$ & 0.42364 & 3 & 0.74016 & $-0.17706$ & 0.42364 \\
		8  & 8.1786 & $-1.3697$ & 3.9544    & --- & --- & --- & --- \\
		9  & 3.0141 & 0.67612 & $-0.030121$ & 4 & 3.0141 & 0.67612 & $-0.030121$ \\
		10 & 12.234 & 2.6613 & 0.080632    & 5 & 12.234 & 2.6613 & 0.080632 \\
		\bottomrule
	\end{tabular}
\end{center}
\end{strip}

Rows $\{1,4,7,9,10\} = \text{EssI}$ match exactly, confirming $\mathbf{V}_{dr} = \mathbf{V}\,\mathbf{Q}_{dr}$ for the orthogonal rotation implicit in $\mathbf{T}$.

$\mathbf{U}$ and $\mathbf{U}_{dc}$: column-space rotation

Analogously, $\mathbf{T} = \mathbf{U}_{dc}\backslash \mathbf{U}$ and $\mathbf{Y}_{dc}^{\mathrm{tr}} := \mathbf{Y}/\mathbf{T}$ reproduce $\mathbf{Y}_{dc}$ exactly on the rows indexed by $\text{EssJ}$:

\begin{strip}
\begin{center}
	\begin{tabular}{@{}c ccc c ccc@{}}
		\toprule
		Row & \multicolumn{3}{c}{$\mathbf{Y}_{dc}^{\mathrm{tr}}$} & Match & \multicolumn{3}{c}{$\mathbf{Y}_{dc}$} \\
		\midrule
		3  & 3.4853 & 2.8071   & $-0.32441$ & 1 & 3.4853 & 2.8071   & $-0.32441$ \\
		5  & 10.65  & $-0.65812$ & $-0.79873$ & 2 & 10.65  & $-0.65812$ & $-0.79873$ \\
		6  & 4.1431 & 2.9576   & $-0.69207$ & 3 & 4.1431 & 2.9576   & $-0.69207$ \\
		9  & 10.938 & $-1.631$ & 0.064216   & 4 & 10.938 & $-1.631$ & 0.064216 \\
		15 & 7.7442 & 0.73692  & $-0.90501$ & 5 & 7.7442 & 0.73692  & $-0.90501$ \\
		17 & 10.146 & $-1.439$ & 1.3406     & 6 & 10.146 & $-1.439$ & 1.3406 \\
		18 & 3.4461 & 3.3961   & 1.5113     & 7 & 3.4461 & 3.3961   & 1.5113 \\
		\bottomrule
	\end{tabular}
\end{center}
\end{strip}

(All 13 remaining, non-essential rows produce no match, as expected, and are omitted here for brevity.) The matched rows are exactly $\text{EssJ} = \{3,5,6,9,15,17,18\}$.

\subsubsection*{$\mathbf{V}_{dc}$ and $\mathbf{V}_{drc}$, $\mathbf{U}_{dr}$ and $\mathbf{U}_{drc}$: the double-reduction consistency}

The same procedure applied to the once-reduced matrices confirms the commuting-square property from the general theory: $\mathbf{T}=\mathbf{V}_{drc}\backslash \mathbf{V}_{dc}$ transforms $\mathbf{X}_{dc}$ into $\mathbf{X}_{drc}$ exactly on the rows indexed by $\text{EssI}$,

\begin{strip}
\begin{center}
	\begin{tabular}{@{}c ccc c ccc@{}}
		\toprule
		Row & \multicolumn{3}{c}{$\mathbf{X}_{drc}^{\mathrm{tr}}$} & Match & \multicolumn{3}{c}{$\mathbf{X}_{drc}$} \\
		\midrule
		1  & 5.6743 & $-1.4045$ & 0.24813   & 1 & 5.6743 & $-1.4045$ & 0.24813 \\
		4  & 5.8223 & $-0.87413$ & $-0.31416$ & 2 & 5.8223 & $-0.87413$ & $-0.31416$ \\
		7  & 0.3751 & $-0.11109$ & 0.27982  & 3 & 0.3751 & $-0.11109$ & 0.27982 \\
		9  & 1.504  & 0.513     & $-0.019028$ & 4 & 1.504  & 0.513     & $-0.019028$ \\
		10 & 6.0873 & 2.0254    & 0.056648 & 5 & 6.0873 & 2.0254    & 0.056648 \\
		\bottomrule
	\end{tabular}
\end{center}
\end{strip}

and, symmetrically, $\mathbf{T} = \mathbf{U}_{drc}\backslash \mathbf{U}_{dr}$ transforms $\mathbf{Y}_{dr}$ into $\mathbf{Y}_{drc}$ exactly on the rows indexed by $\text{EssJ}$:

\begin{strip}
\begin{center}
	\begin{tabular}{@{}c ccc c ccc@{}}
		\toprule
		Row & \multicolumn{3}{c}{$\mathbf{Y}_{drc}^{\mathrm{tr}}$} & Match & \multicolumn{3}{c}{$\mathbf{Y}_{drc}$} \\
		\midrule
		3  & 4.0188 & 0.085184  & $-0.173$   & 1 & 4.0188 & 0.085184  & $-0.173$ \\
		5  & 4.0521 & $-1.1424$ & 0.095101   & 2 & 4.0521 & $-1.1424$ & 0.095101 \\
		6  & 4.4854 & $-0.26415$ & $-0.20604$ & 3 & 4.4854 & $-0.26415$ & $-0.20604$ \\
		9  & 3.1844 & $-0.47178$ & 0.23369   & 4 & 3.1844 & $-0.47178$ & 0.23369 \\
		15 & 4.0732 & $-0.93877$ & $-0.037652$ & 5 & 4.0732 & $-0.93877$ & $-0.037652$ \\
		17 & 2.8318 & 0.8295    & 0.31787   & 6 & 2.8318 & 0.8295    & 0.31787 \\
		18 & 4.2562 & 1.9851    & $-0.060355$ & 7 & 4.2562 & 1.9851    & $-0.060355$ \\
		\bottomrule
	\end{tabular}
\end{center}
\end{strip}

This case study confirms, on real chemometric data, exactly the structural results proved in general: the row space (here, the space of element-loadings $\mathbf{V}$) is rotated but not distorted when non-essential \emph{samples} are removed, and the column space ($\mathbf{U}$) is rotated but not distorted when non-essential \emph{elements} are removed. The Essential Data Points --- the convex-hull vertices of $\mathbf{X}$ and $\mathbf{Y}$ --- are precisely the rows/columns that survive this rotation exactly; every non-essential point fails to match because it was, by construction, expressible as a combination of the essential ones and carries no independent information once removed.

The other subspace comparisons are not possible by this naive way, because there is not any simple executable Matlab command to get a proper rotation matrix similar to $\mathbf{T}$s, since the incompatible matrices.
It can be illustrated using the Matlab command \verb+subspace+  for comparison of the compatible cases: \verb+subspace+(V, Vdr) = $1.4182\cdot10^{-15}$, \verb+subspace+(U, Udc) = $8.9883\cdot10^{-16}$, \verb+subspace+(Vdc, Vdrc) = $4.5371\cdot10^{-16}$, \verb+subspace+(Udr, Udrc) = $4.5809\cdot10^{-16}$. However the commands of \verb+subspace+(Vdr, Vdc), \verb+subspace+(Udr, Udc), \verb+subspace+(Vdr, Vdrc), \verb+subspace+(Udc, Udrc) result in error message: \textsl{"Incorrect dimensions for matrix multiplication. Check that the number of columns in the first matrix matches the number of rows in the second matrix"}. Fortunately, based on the duality between the row-space and column-space of every matrix: $\mathfrak{Col}(A^\mathsf{T}) = \mathfrak{Row}(A)$. Thus \verb+subspace+(Vdr$'$, Vdc$'$) = $3.5928\cdot10^{-16}$, \verb+subspace+(Udr$'$, Udc$'$) = $0$, \verb+subspace+(Vdr$'$, Vdrc$'$) = $3.3990\cdot10^{-16}$, \verb+subspace+(Udc$'$, Udrc$'$) = $3.8589\cdot10^{-16}$.
We can also apply the sophisticated derivations above, which are generally rigorous, see the next subsection.

\subsection{Illustration in strict way}

To confirm the subspace-equivalence theory developed in the main text in strict way, the script in Section~\ref{App:Sec-4-fundSubSpStrct} Appendix was run on the same real source-apportionment matrix $\mathbf{D} = \mathbf{C}\,\mathbf{A}^\mathsf{T} \in \mathbb{R}^{10\times20}$ (rank $r=3$) rather than on synthetic data. The Essential Data Points --- $5$ of the $10$ rows and $7$ of the $20$ columns --- were kept fixed at the convex-hull vertex sets $\mathcal{I}=\{1,4,7,9,10\}$ and $\mathcal{J}=\{3,5,6,9,15,17,18\}$ identified previously, and every identity derived analytically was checked against the corresponding numerical residual. All results below are reported as Frobenius-norm residuals, which should vanish to floating-point precision ($\sim 10^{-14}$--$10^{-16}$) whenever the underlying identity is exact.

\subsubsection*{Type 1: Pure rotations}

These checks confirm that the row space of $\mathbf{D}$ is preserved \emph{exactly} under row deletion, and the column space of $\mathbf{D}$ is preserved \emph{exactly} under column deletion, in both cases via an orthogonal rotation.

\begin{strip}
\begin{center}
	\begin{tabular}{@{}l l S[table-format=1.3e2]@{}}
		\toprule
		Quantity & Identity checked & {Residual} \\
		\midrule
		$\mathbf{Q}_{dr}=\mathbf{V}^\mathsf{T}\,\mathbf{V}_{dr}$ orthogonality & $\|\mathbf{Q}_{dr}^\mathsf{T}\,\mathbf{Q}_{dr}-\mathbf{I}\|$ & 9.612e-16 \\
		Row space preserved & $\|\mathbf{V}\,\mathbf{Q}_{dr}-\mathbf{V}_{dr}\|_F$ & 1.849e-15 \\
		$\mathbf{Q}_{dc}=\mathbf{U}^\mathsf{T}\,\mathbf{U}_{dc}$ orthogonality & $\|\mathbf{Q}_{dc}^\mathsf{T}\,\mathbf{Q}_{dc}-\mathbf{I}\|$ & 1.337e-15 \\
		Column space preserved & $\|\mathbf{U}\,\mathbf{Q}_{dc}-\mathbf{U}_{dc}\|_F$ & 9.479e-16 \\
		\bottomrule
	\end{tabular}
\end{center}
\end{strip}

\subsubsection*{Type 2: Ambient-shrinking isomorphisms}

These checks confirm that on the side where the ambient space actually shrinks (rows for $\mathbf{U}$, columns for $\mathbf{V}$), the relevant map is a general invertible transform, \emph{not} an orthogonal one.

\begin{strip}
\begin{center}
	\begin{tabular}{@{}l l S[table-format=1.3e2]@{}}
		\toprule
		Quantity & Identity checked & {Residual} \\
		\midrule
		$\mathbf{U}(\mathcal{I},:) = \mathbf{U}_{dr}N_{dr}$ & $\|\mathbf{U}(\mathcal{I},:) - \mathbf{U}_{dr}\,\mathbf{N}_{dr}\|_F$ & 2.523e-16 \\
		$\mathbf{V}(\mathcal{J},:) = \mathbf{V}_{dc}\,\mathbf{M}$ & $\|\mathbf{V}(\mathcal{J},:) - \mathbf{V}_{dc}\,\mathbf{M}\|_F$ & 4.983e-16 \\
		$\mathbf{N}_{dr}$ orthogonality (expected to fail) & $\|\mathbf{N}_{dr}^\mathsf{T}\,\mathbf{N}_{dr}-\mathbf{I}\|$ & \textbf{9.980e-01} \\
		\bottomrule
	\end{tabular}
\end{center}
\end{strip}

The last row is the key confirmation: $N_{dr}$ is emphatically \emph{not} orthogonal (residual $\approx 1$, not $\approx 10^{-15}$), verifying that the ambient-shrinking map is a genuine isomorphism rather than a rotation, exactly as predicted.

\subsubsection*{Derived clean relations for $\mathbf{X}$ and $\mathbf{Y}$}

Although $\mathbf{U}(\mathcal{I},:)$ and $\mathbf{V}(\mathcal{J},:)$ transform non-orthogonally, the $\mathbf{S}$-factors cancel exactly when forming $\mathbf{X}=\mathbf{U}\,\mathbf{S}$ and $\mathbf{Y}=\mathbf{V}\,\mathbf{S}$, restoring a clean rotation on the \emph{restricted} factor:

\begin{strip}
\begin{center}
	\begin{tabular}{@{}l l S[table-format=1.3e2]@{}}
		\toprule
		Quantity & Identity checked & {Residual} \\
		\midrule
		$\mathbf{X}_{dr} = \mathbf{X}(\mathcal{I},:)Q_{dr}$ & $\|\mathbf{X}_{dr} - \mathbf{X}(\mathcal{I},:)\,\mathbf{Q}_{dr}\|_F$ & 1.102e-14 \\
		$\mathbf{Y}_{dc} = \mathbf{Y}(\mathcal{J},:)\,\mathbf{Q}_{dc}$ & $\|\mathbf{Y}_{dc} - \mathbf{Y}(\mathcal{J},:)\,\mathbf{Q}_{dc}\|_F$ & 1.464e-14 \\
		\bottomrule
	\end{tabular}
\end{center}
\end{strip}

\subsubsection*{Type 3: $\boldsymbol{\Sigma}$-mismatch isomorphisms}

These checks confirm that even where the ambient space is \emph{not} shrinking (full $X$ under column deletion, full $\mathbf{Y}$ under row deletion), the transform is still a general invertible matrix rather than a rotation, because the singular values themselves differ ($\boldsymbol{\Sigma} \neq \boldsymbol{\Sigma}_{dc}$, $\boldsymbol{\Sigma} \neq \boldsymbol{\Sigma}_{dr}$):

\begin{strip}
\begin{center}
	\begin{tabular}{@{}l l S[table-format=1.3e2]@{}}
		\toprule
		Quantity & Identity checked & {Residual} \\
		\midrule
		$\mathbf{X}_{dc} = \mathbf{X}\,\mathbf{K}_{dc}$ & $\|\mathbf{X}_{dc} - \mathbf{X}\,\mathbf{K}_{dc}\|_F$ & 6.273e-15 \\
		$\mathbf{K}_{dc}$ closed form & $\|\mathbf{K}_{dc} - \mathbf{S}^{-1}\,\mathbf{Q}_{dc}S_{dc}\|_F$ & 7.069e-16 \\
		$\mathbf{Y}_{dr} = \mathbf{Y}\,\mathbf{L}_{dr}$ & $\|\mathbf{Y}_{dr} - \mathbf{Y}\,\mathbf{L}_{dr}\|_F$ & 5.336e-15 \\
		$\mathbf{L}_{dr}$ closed form & $\|\mathbf{L}_{dr} - \mathbf{S}^{-1}\,\mathbf{Q}_{dr}\,\mathbf{S}_{dr}\|_F$ & 6.188e-16 \\
		\bottomrule
	\end{tabular}
\end{center}
\end{strip}

\subsubsection*{Double deletion: path-independence}

Finally, the double-deletion case confirms both that the row space of $\mathbf{D}(:,\mathcal{J})$ is preserved exactly under further row deletion, and that composing the column-deletion isomorphism with this rotation reproduces $\mathbf{X}_{drc}$ exactly:

\begin{strip}
\begin{center}
	\begin{tabular}{@{}l l S[table-format=1.3e2]@{}}
		\toprule
		Quantity & Identity checked & {Residual} \\
		\midrule
		$\mathbf{Q}_{drc}=\mathbf{V}_{dc}^\mathsf{T}\,\mathbf{V}_{drc}$ orthogonality & $\|\mathbf{Q}_{drc}^\mathsf{T}\,\mathbf{Q}_{drc}-\mathbf{I}\|$ & 9.905e-16 \\
		Subspace preserved & $\|\mathbf{V}_{dc}\,\mathbf{Q}_{drc}-\mathbf{V}_{drc}\|_F$ & 1.321e-15 \\
		Composed transform & $\|\mathbf{X}_{drc} - \mathbf{X}(\mathcal{I},:)\,\mathbf{M}^\mathsf{T}\,\mathbf{Q}_{drc}\|_F$ & 1.218e-14 \\
		\bottomrule
	\end{tabular}
\end{center}
\end{strip}

\subsection*{Summary of Findings}

The rank check confirms $\operatorname{rank}(\mathbf{D}_{dr}) = \operatorname{rank}(\mathbf{D}_{dc}) = \operatorname{rank}(\mathbf{D}_{drc}) = r = 3$, as required for a genuine convex-hull-vertex (rank-preserving) deletion. Table collects the qualitative conclusion for every quantity examined:

\begin{strip}
\begin{center}
	\begin{tabular}{@{}l l@{}}
		\toprule
		Pair & Relationship confirmed numerically \\
		\midrule
		$\mathbf{V},\ \mathbf{V}_{dr}$ & Identical subspace of $\mathbb{R}^n$ (rotation $\mathbf{Q}_{dr}$) \\
		$\mathbf{U},\ \mathbf{U}_{dc}$ & Identical subspace of $\mathbb{R}^m$ (rotation $\mathbf{Q}_{dc}$) \\
		$\mathbf{U}(\mathcal{I},:),\ \mathbf{U}_{dr}$ & Isomorphic; ambient shrinks $m\to p$ (map $\mathbf{N}_{dr}$, not orthogonal) \\
		$\mathbf{V}(\mathcal{J},:),\ \mathbf{V}_{dc}$ & Isomorphic; ambient shrinks $n\to q$ (map $\mathbf{M}$, not orthogonal) \\
		$\mathbf{X},\ \mathbf{X}_{dc}$ & Related by invertible $\mathbf{K}_{dc}$ ($\boldsymbol{\Sigma}$ changes, ambient fixed) \\
		$\mathbf{Y},\ \mathbf{Y}_{dr}$ & Related by invertible $\mathbf{L}_{dr}$ ($\boldsymbol{\Sigma}$ changes, ambient fixed) \\
		$\mathbf{X}_{dr},\ \mathbf{X}(\mathcal{I},:)$ & Clean rotation $\mathbf{Q}_{dr}$ (always exact) \\
		$\mathbf{Y}_{dc},\ \mathbf{Y}(\mathcal{J},:)$ & Clean rotation $\mathbf{Q}_{dc}$ (always exact) \\
		$\mathbf{V}_{drc},\ \mathbf{X}_{drc}$ & Obtained by composing the column isomorphism with a further rotation \\
		\bottomrule
	\end{tabular}
	\label{tab:summary}
\end{center}
\end{strip}

Every residual reported above is at or below $\sim 10^{-14}$, except the deliberately-negative control ($\mathbf{N}_{dr}$ orthogonality, $\approx 1$), which confirms --- rather than contradicts --- the theory: it demonstrates that the ambient-shrinking maps are genuinely non-orthogonal, as the analytical derivation predicted.

\section{Conclusions}

To evaluate the effect of essential-data-point (EDP) based data reduction on
the geometry of the resulting polygons, the inner polygons obtained from the
reduced data sets were compared with the corresponding polygon of the
original data set (Fig.~\ref{fig9:disortedInOutPols}).
Data reduction was found to alter not only the
coordinates of the polygon but also its orientation relative to the polygon
obtained from the original data set: the inner polygon of a reduced data set
appears \emph{rotated} with respect to that of the original.

To determine whether this apparent rotation is a genuine orthogonal rotation
or a more general, non-orthogonal isomorphism, the transformation matrix
between the original and each reduced data set was estimated by least
squares, e.g.,
$
\mathbf{V}=\mathbf{V}_{dr}\, \mathbf{T}_{drv} \;\Rightarrow\; \mathbf{T}_{drv}=\mathbf{V}^\mathsf{T}_{dr}\,\mathbf{V},
\quad
\mathbf{U}=\mathbf{U}_{dr}\, \mathbf{T}_{dru} \;\Rightarrow\; \mathbf{T}_{dru}=\mathbf{U}^\mathsf{T}_{dr}\,\mathbf{U}.
$
For an orthogonal transformation, $\mathbf{T^\mathsf{T}}\mathbf{T} = \mathbf{I}$ must hold exactly. This
check was carried out for all three data-reduction scenarios (row-wise,
column-wise, jointly) in both the row and column spaces (Fig.~\ref{fig10:TransformatMatirces}).

\textbf{Main finding.} The results confirm, both qualitatively (Fig.~\ref{fig9:disortedInOutPols},
Fig.~\ref{fig10:TransformatMatirces}) and rigorously (Section~\ref{subsec:subspEqunderCHdel}), a single structural rule: of the two
transformation matrices associated with any EDP reduction, exactly the ones
belonging to the \emph{non-reduced} mode are orthogonal ($\mathbf{T^\mathsf{T}}\mathbf{T} = \mathbf{I}$ to
machine precision, $\sim 10^{-15}$--$10^{-16}$ in our numerical tests),
while those belonging to the \emph{reduced} mode are demonstrably
non-orthogonal (residual $\approx 1$ rather than $\approx 0$ in our negative
control). In row-wise reduction, the row-space transformation is an exact
rotation while the column-space transformation is a non-orthogonal
isomorphism; in column-wise reduction the roles reverse; and in joint
row-and-column reduction, neither transformation is guaranteed to be
orthogonal. Consequently, the geometric changes induced by EDP-based removal
of non-essential rows and/or columns \emph{cannot, in general, be described
as an orthogonal rotation}, and treating them as such --- e.g., by
superimposing polygons without first estimating and applying the correct
transformation --- will misrepresent the true relationship between the
reduced and original data geometry.

\textbf{Practical recommendation.} Based on the investigations above, the
central, actionable conclusion of this work is: \emph{if EDPs are used for
any subsequent interpretation of the abstract row or column space
(e.g., \cite{Beyramysoltan2021,KHODADADIKARIMVAND2023,Coic2023,OLARINI2024,QING2024,})}, the correct mode-specific
transformation ($\mathbf{T}_v$ or $\mathbf{T}_u$, orthogonal or not) must be estimated and
applied \emph{before} comparing the reduced and original representations.
Skipping this step and comparing raw coordinates, or assuming the
discrepancy is a rigid rotation removable by a generic Procrustes fit, can
lead to a false impression of either preserved or lost structure.

\textbf{Limitations and future work.} The present analysis is based on
noise-free simulated data and a single real source-apportionment data set;
the behavior of the transformation matrices under realistic measurement
noise, and their sensitivity to the specific convex-hull algorithm used to
identify EDPs, remain to be characterized. Extending the present
subspace-equivalence framework to three-way (\cite{Vitale2024}) and higher-order data structures,
and deriving explicit error bounds on the non-orthogonal isomorphisms
($\mathbf{N}_{dr}$, $\mathbf{M}$) as a function of how far a deleted point lies from the
convex hull boundary, would further clarify when EDP-based reduction is
safe to use without correction and when it is not.

\begin{figure}
	\centering
	\includegraphics[width=0.45\textwidth]{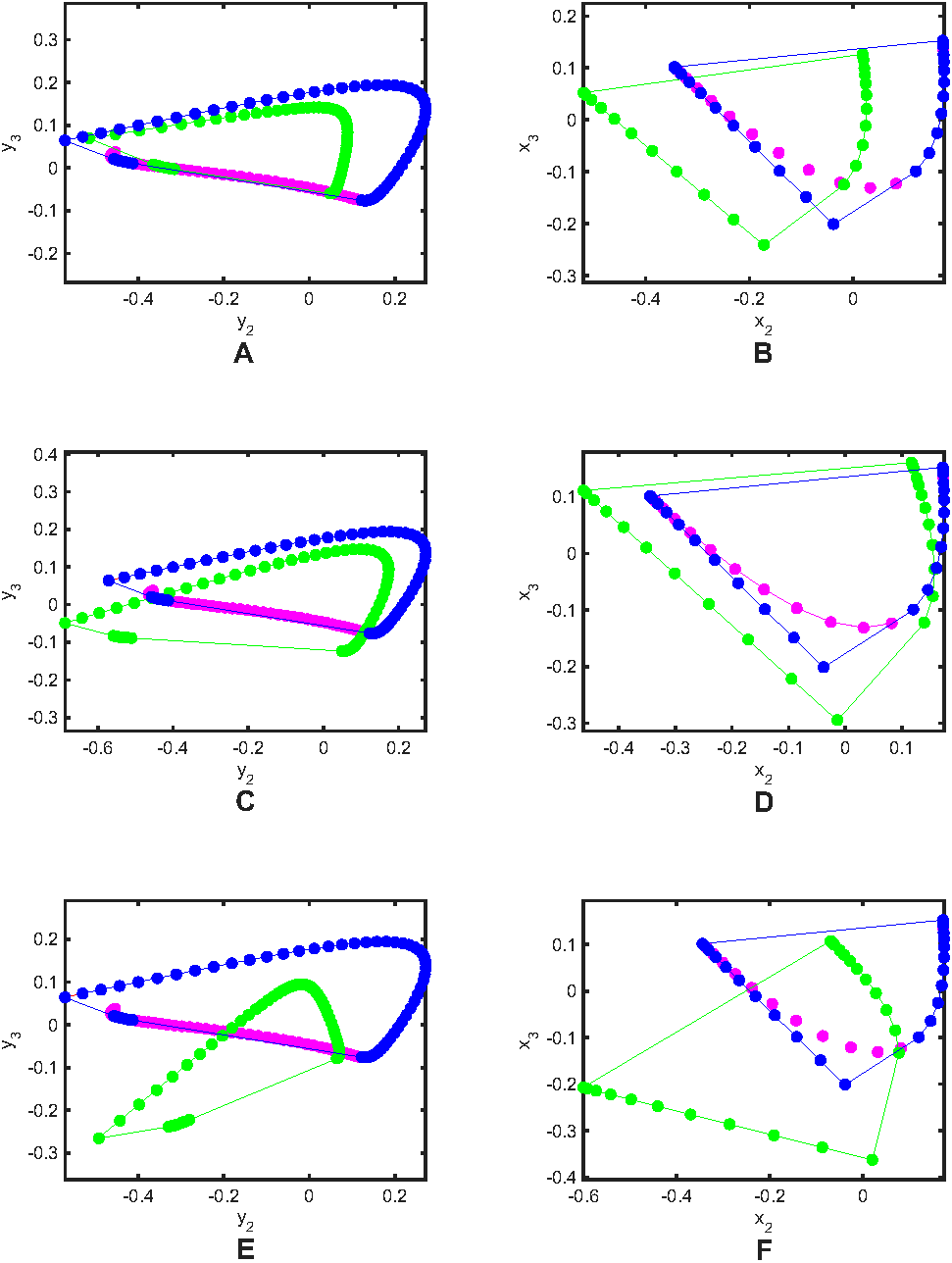} \\
	\caption{Comparison of the inner polygons obtained from the SVD of the original and reduced datasets in the $V$- and $U$-spaces. \textbf{First row}: comparison between the original dataset ($\mathbf{D}$) and the row-reduced dataset ($\mathbf{D}_{dr}$); (A) comparison of the inner polygons obtained from the $V$-space and (B) comparison of the inner polygons obtained from the $U$-space. \textbf{Second row}: comparison between ($\mathbf{D}$) and the column-reduced dataset ($\mathbf{D}_{dc}$); (C) comparison of the inner polygons obtained from the $V$-space and (D) comparison of the inner polygons obtained from the $U$-space. \textbf{Third row}: comparison between ($\mathbf{D}$) and the dataset reduced in both rows and columns, ($\mathbf{D}_{drc}$); (E) comparison of the inner polygons obtained from the $V$-space and (F) comparison of the inner polygons obtained from the $U$-space. In all panels, the red polygons represent the original dataset, whereas the green polygons represent the corresponding reduced dataset. The blue polygons represent the corrected reduced dataset. Notably, the blue and red polygons completely overlap, indicating that the corrected reduced dataset reproduces the geometry of the original dataset. The pink points correspond to the non-essential data points.}
	\label{fig9:disortedInOutPols}
\end{figure}

\begin{figure}
	\centering
	\includegraphics[width=0.35\textwidth]{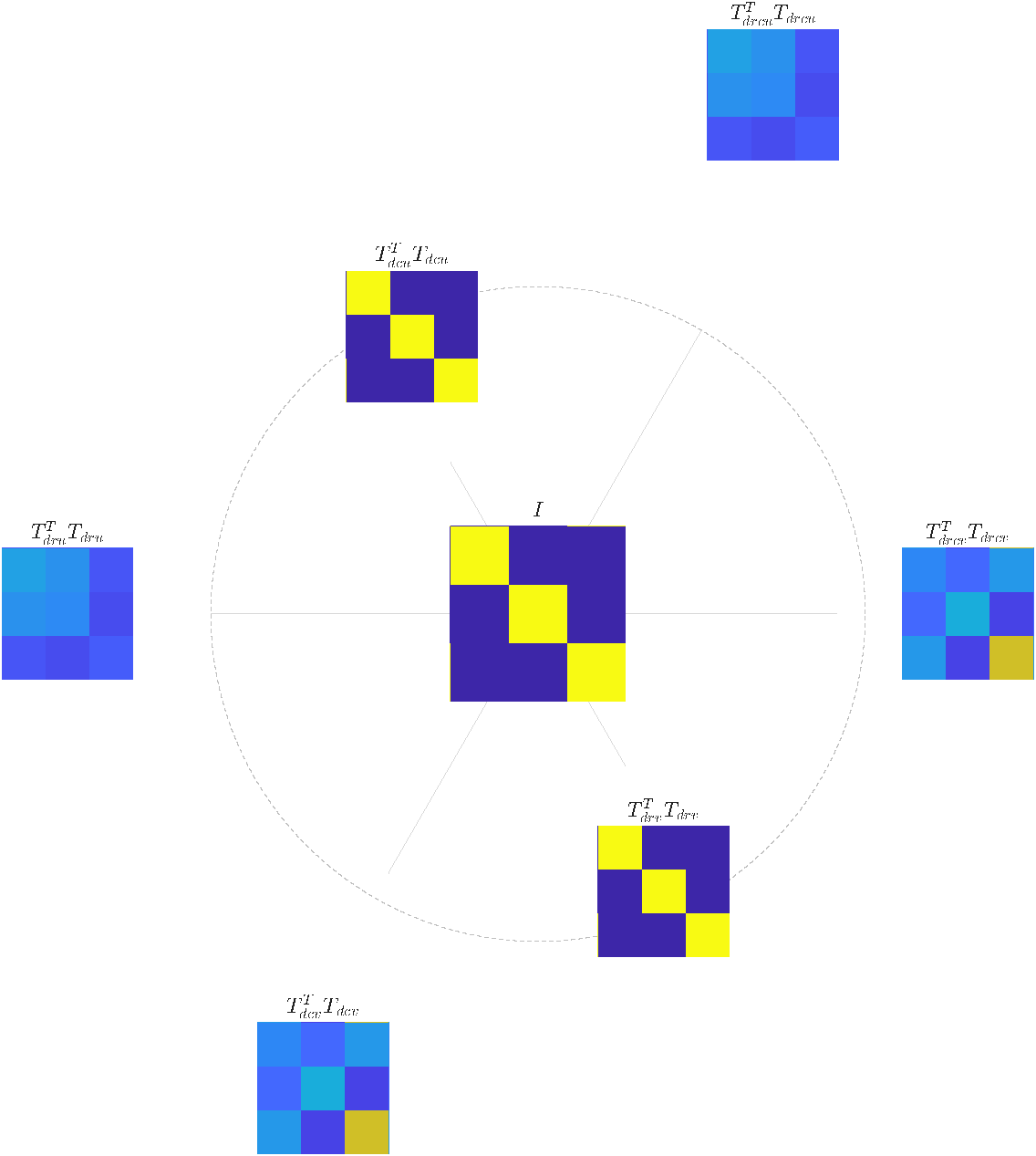} \\
	\caption{Heatmaps of the transformation matrices obtained by least-squares comparison of the SVD-derived $\mathbf{U}$ and $\mathbf{V}$ matrices between the original dataset $\mathbf{D}$ and the corresponding reduced datasets. The six transformation matrices correspond to the comparisons between $\mathbf{D}$ and $\mathbf{D}_{dr}$ in the $U$- and $V$-spaces, between $\mathbf{D}$ and $\mathbf{D}_{dc}$ in the $U$- and $V$-spaces, and between $\mathbf{D}$ and $\mathbf{D}_{drc}$ in the $U$- and $V$-spaces, respectively. Among these six transformations, two are orthogonal whereas the remaining four are non-orthogonal.}
	\label{fig10:TransformatMatirces}
\end{figure}

\section{Appendix}

\subsection{Short introduction to the Strang diagram of the four fundamental subspaces}\label{App:Sec-4-fundSubSp}
Fig.~\ref{fig11:4fundspsc} depicts the classic Strang diagram~\cite{strang2023}, showing exactly how a linear map $A:\mathbb{C}^n \to \mathbb{C}^m$ routes every input vector to its output. Here's a detailed walkthrough using the picture's own notation.

\begin{figure}
	\centering
	\includegraphics[width=0.45\textwidth]{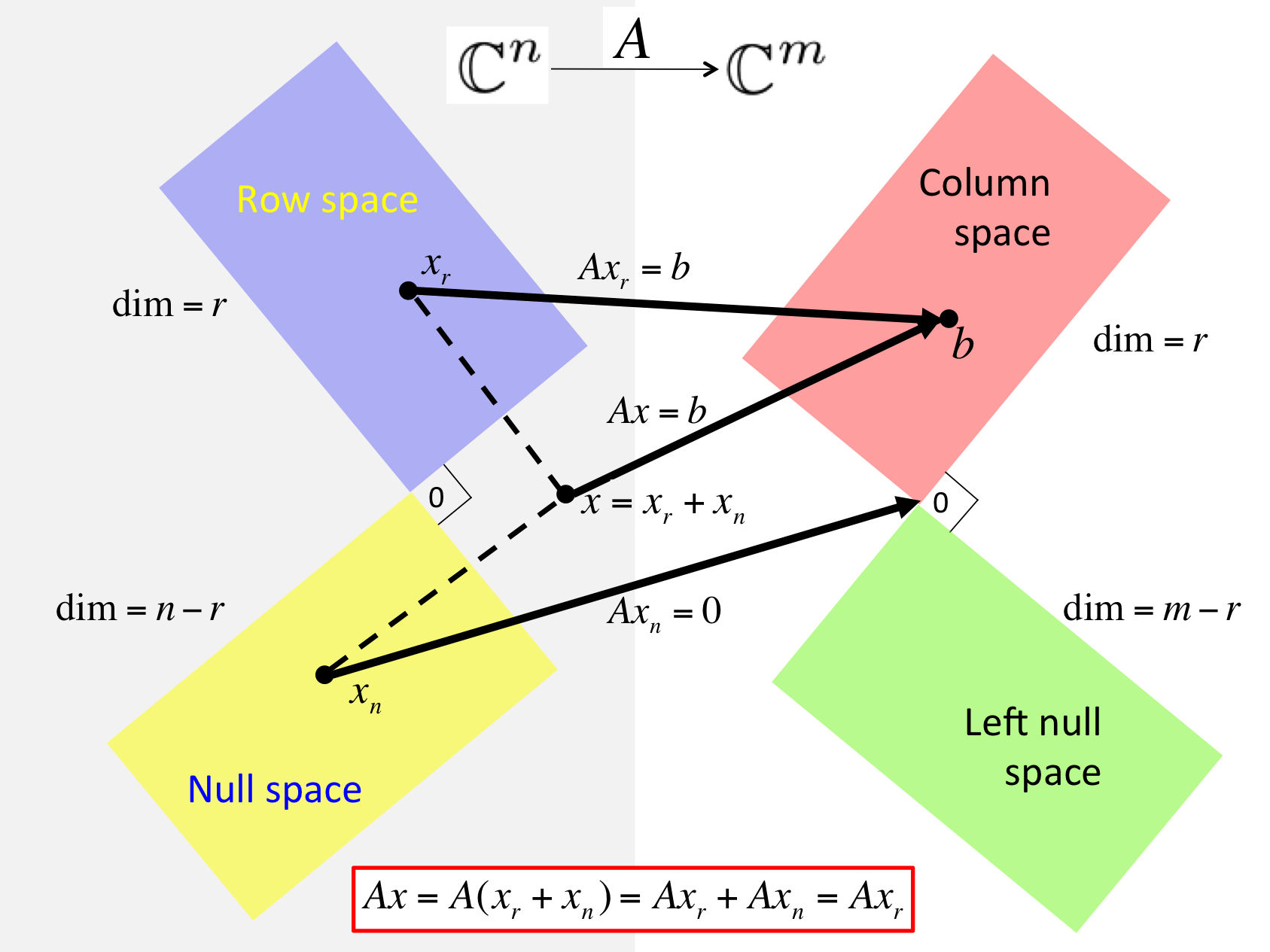} \\
	\caption{Illustration of the four fundamental subspaces and the mapping of a vector $x\in \mathbb{C}^n$ by matrix $A \in \mathbb{C}^{m\times n}$. (Source: \url{https://www.cs.utexas.edu/~flame/laff/alaff/chapter04-four-fundamental-spaces.html}) }
	\label{fig11:4fundspsc}
\end{figure}

\subsubsection*{The two domains}

\textbf{Left side ($\mathbb{C}^n$):} the input/domain space. It splits into two orthogonal pieces:
\begin{itemize}
	\item \textbf{Row space} (blue, top-left), dimension $r$ --- this is $\mathfrak{Col}(A^\mathsf{T})$, the space spanned by the rows of $A$.
	\item \textbf{Null space} (yellow, bottom-left), dimension $n-r$ --- this is $\mathfrak{Null}(A)$, the set of vectors $A$ annihilates.
\end{itemize}
These two occupy complementary dimensions of $\mathbb{C}^n$ and together account for all $n$ dimensions: $r + (n-r) = n$. 
\textbf{Right side ($\mathbb{C}^m$):} the output/codomain space. It also splits into two orthogonal pieces:
\begin{itemize}
	\item \textbf{Column space} (red, top-right), dimension $r$ --- this is $\mathfrak{Col}(A)$, the space spanned by the columns of $A$, i.e.\ every vector $A$ can actually produce.
	\item \textbf{Left null space} (green, bottom-right), dimension $m-r$ --- this is $\mathfrak{Null}(A^\mathsf{T})$, the set of output directions $A$ can never reach.
\end{itemize}
Again these complement each other: $r + (m-r) = m$.

\subsubsection*{What the arrows show}
Take any input vector $x \in \mathbb{C}^n$. The diagram decomposes it uniquely as $x = x_r + x_n$, 
where $x_r$ lies in the row space (blue) and $x_n$ lies in the null space (yellow) --- the dashed lines show this splitting, meeting at the origin $0$ on the left.

\begin{itemize}
	\item The arrow labeled $Ax_n = 0$ shows that the null-space component is completely destroyed: it maps straight to the origin $0$ on the right side (in the left-null-space corner), contributing nothing to the output.
	\item The arrow labeled $Ax_r = b$ shows that the row-space component alone determines the entire output: $x_r$ maps into the column space (red) and lands exactly at $b$.
	\item The arrow labeled $Ax = b$ is simply the full map applied to the original $x$ --- and it lands at the \emph{same} point $b$ as $Ax_r$ does, because the null-space part contributed zero.
\end{itemize}

\subsubsection*{The boxed identity}
\[
\boxed{Ax = A(x_r + x_n) = Ax_r + Ax_n = Ax_r}
\]
This equation is the entire content of the picture in one line: applying $A$ to any vector only ever depends on that vector's row-space component. The null-space component is invisible to the transformation --- it is thrown away --- while the row-space component alone is responsible for generating the observed output $b$ inside the column space.

\subsubsection*{Why the picture is drawn this way}

\begin{itemize}
	\item $A$ acts as a genuine \textbf{invertible correspondence} between the row space (dimension $r$) and the column space (dimension $r$) --- this is the ``one-to-one, onto'' restriction of $A$ that survives once the null space is quotiented out. That is why the diagram draws $Ax_r = b$ as a clean, direct arrow between those two same-dimensional blue and red boxes.
	\item The null space and left null space (yellow, green) are the two ``dead zones'': one on the input side (things that vanish), one on the output side (places nothing can go). The picture places them at the bottom to emphasize that they are orthogonal complements of the row space and column space respectively, connected to $A$ only through the trivial map $x_n \mapsto 0$.
\end{itemize}

This single diagram encodes rank-nullity on both sides, $r + (n-r) = n, r + (m-r) = m$ 
and the fact that $A$ restricted to the row space is a bijection onto the column space --- the geometric essence of the rank of a matrix.

\subsection{MATLAB Code: Essential Data Points (EDPs) Naive Verification}
\label{app:matlab-edp-naive}

The following code reproduces, on the Henry \& Kim (1990) source-apportionment dataset~\cite{HENRY1990}, the four subspace-equivalence checks discussed in the main text: it identifies the Essential Data Points (convex-hull vertices) of the row-points and column-points, builds the reduced matrices $R_{dr}$, $R_{dc}$, $R_{drc}$, and numerically verifies that $V \leftrightarrow V_{dr}$, $U \leftrightarrow U_{dc}$, $V_{dc}\leftrightarrow V_{drc}$, and $U_{dr}\leftrightarrow U_{drc}$ are related by exact orthogonal rotations on the essential rows/columns.

\subsubsection*{Main script}
\begin{strip}
\begin{lstlisting}[style=matlabstyle, caption={Main analysis script}, label={lst:edp-main}]
	% Essential Data Points (EDPs) are correct
	clearvars
	
	% data from the paper of Henry & Kim, ChemoLab, 8 (1990) 205-216
	% https://doi.org/10.1016/0169-7439(90)80136-T
	
	% Source compositions
	%      Marine    Udust     Auto
	C = [[ 0.40000  0.01250  0.00000]; ...   % Na
	     [ 0.00000  0.08840  0.01100]; ...   % Al
	     [ 0.00000  0.22300  0.00820]; ...   % Si
	     [ 0.40000  0.00000  0.03000]; ...   % Cl
	     [ 0.01400  0.01030  0.00072]; ...   % K
	     [ 0.01400  0.02440  0.01250]; ...   % Ca
	     [ 0.00000  0.00640  0.00000]; ...   % Ti
	     [ 0.00000  0.06000  0.02100]; ...   % Fe
	     [ 0.00200  0.00020  0.05000]; ...   % Br
	     [ 0.00000  0.00370  0.20000]];      % Pb
	
	% Source apportionment ug/m3
	%      Marine    Udust     Auto
	A = [[   3         8        19]; ...
	     [   3         8        16]; ...
	     [   5         7        12]; ...
	     [   6        49        13]; ...
	     [   6        39         7]; ...
	     [   6         9        12]; ...
	     [   6        17        18]; ...
	     [   2         6         6]; ...
	     [   4        42         7]; ...
	     [   4        49        20]; ...
	     [   4        29        14]; ...
	     [   5        44         9]; ...
	     [   5        32        12]; ...
	     [   5         9        11]; ...
	     [   6        26         8]; ...
	     [   3         8        14]; ...
	     [   2        39        11]; ...
	     [   3         6        20]; ...
	     [   4        48        15]; ...
	     [   5        37        14]];
	
	ElemsTxt = {'Na'; 'Al'; 'Si'; 'Cl'; 'K'; 'Ca'; 'Ti'; 'Fe'; 'Br'; 'Pb'};
	ApporTxt = {'Marine'; 'Udust'; 'Auto'};
	
	% Convex Hull (CH) indices for inner polygons
	% indices for inner polygons = Essential Data Points (EDPs)
	IndsH3.XnInn = [
	1     4
	4     9
	9    10
	10    7
	7     1];
	EssI = unique(IndsH3.XnInn);
	
	IndsH3.YnInn = [
	5    15
	15    6
	6     3
	3    18
	18   17
	17    9
	9     5];
	EssJ = unique(IndsH3.YnInn);
	
	epsi = 1e-15;
	
	%% Full data matrix and its SVD-based loadings
	R = C*A'; R_dual = R';
	[I, J] = size(R);
	
	[Um, Sm, Vm] = svds(R,3);
	loads{1} = Um*Sm; loads{2} = Vm;
	[sgnsXV, loadsXV] = sign_flip(loads, R);
	
	loads{1} = Vm*Sm; loads{2} = Um;
	[sgnsYU, loadsYU] = sign_flip(loads, R_dual);
	
	X = loadsXV{1}; V = loadsXV{2};
	Y = loadsYU{1}; U = loadsYU{2};
	
	%% Row-deleted reduction (keep essential rows I, all columns)
	C_ = C(EssI,:); A_ = A; ElemsTxt_ = ElemsTxt(EssI,:); ApporTxt_ = ApporTxt;
	R_ = C_*A_'; R_dual_ = R_';
	
	[Um, Sm, Vm] = svds(R_,3);
	loads{1} = Um*Sm; loads{2} = Vm;
	[sgnsXV, loadsXV] = sign_flip(loads, R_);
	
	loads{1} = Vm*Sm; loads{2} = Um;
	[sgnsYU, loadsYU] = sign_flip(loads, R_dual_);
	
	X_ = loadsXV{1}; V_ = loadsXV{2};
	Y_ = loadsYU{1}; U_ = loadsYU{2};
	Xdr = X_; Ydr = Y_; Vdr = V_; Udr = U_;
	
	%% Column-deleted reduction (all rows, keep essential columns J)
	C_ = C; A_ = A(EssJ,:); ElemsTxt_ = ElemsTxt; ApporTxt_ = ApporTxt;
	R_ = C_*A_'; R_dual_ = R_';
	
	[Um, Sm, Vm] = svds(R_,3);
	loads{1} = Um*Sm; loads{2} = Vm;
	[sgnsXV, loadsXV] = sign_flip(loads, R_);
	
	loads{1} = Vm*Sm; loads{2} = Um;
	[sgnsYU, loadsYU] = sign_flip(loads, R_dual_);
	
	X_ = loadsXV{1}; V_ = loadsXV{2};
	Y_ = loadsYU{1}; U_ = loadsYU{2};
	Xdc = X_; Ydc = Y_; Vdc = V_; Udc = U_;
	
	%% Row- and column-deleted reduction (keep essential rows I and columns J)
	C_ = C(EssI,:); A_ = A(EssJ,:);
	ElemsTxt_ = ElemsTxt(EssI,:); ApporTxt_ = ApporTxt;
	R_ = C_*A_'; R_dual_ = R_';
	
	[Um, Sm, Vm] = svds(R_,3);
	loads{1} = Um*Sm; loads{2} = Vm;
	[sgnsXV, loadsXV] = sign_flip(loads, R_);
	
	loads{1} = Vm*Sm; loads{2} = Um;
	[sgnsYU, loadsYU] = sign_flip(loads, R_dual_);
	
	X_ = loadsXV{1}; V_ = loadsXV{2};
	Y_ = loadsYU{1}; U_ = loadsYU{2};
	Xdrc = X_; Ydrc = Y_; Vdrc = V_; Udrc = U_;
	
	%% Verification 1: equivalence of the subspaces of V and Vdr
	disp(' ')
	disp('Equivalence of the subspaces of V and Vdr')
	T = Vdr\V;
	Xdr_tr = X/T;
	[tf, loc] = ismember(round(Xdr_tr,10), round(Xdr,10), 'rows');
	Xdr_ext = NaN(size(Xdr_tr));
	Xdr_ext(logical(loc),:) = Xdr;
	ResultTable = table((1:size(Xdr_tr,1))', Xdr_tr, loc, Xdr_ext, ...
	'VariableNames', {'A_Row_Index', 'Matrix_Xdr_tr', 'B_Row_Index', 'Matrix_Xdr'});
	ResultTable.B_Row_Index(~tf) = NaN;
	disp(ResultTable);
	
	%% Verification 2: equivalence of the subspaces of U and Udc
	disp(' ')
	disp('Equivalence of the subspaces of U and Udc')
	T = Udc\U;
	Ydc_tr = Y/T;
	[tf, loc] = ismember(round(Ydc_tr,10), round(Ydc,10), 'rows');
	Ydc_ext = NaN(size(Ydc_tr));
	Ydc_ext(logical(loc),:) = Ydc;
	ResultTable = table((1:size(Ydc_tr,1))', Ydc_tr, loc, Ydc_ext, ...
	'VariableNames', {'A_Row_Index', 'Matrix_Ydc_tr', 'B_Row_Index', 'Matrix_Ydc'});
	ResultTable.B_Row_Index(~tf) = NaN;
	disp(ResultTable);
	
	%% Verification 3: equivalence of the subspaces of Vdc and Vdrc
	disp(' ')
	disp('Equivalence of the subspaces of Vdc and Vdrc')
	T = Vdrc\Vdc;
	Xdrc_tr = Xdc/T;
	[tf, loc] = ismember(round(Xdrc_tr,10), round(Xdrc,10), 'rows');
	Xdrc_ext = NaN(size(Xdrc_tr));
	Xdrc_ext(logical(loc),:) = Xdrc;
	ResultTable = table((1:size(Xdrc_tr,1))', Xdrc_tr, loc, Xdrc_ext, ...
	'VariableNames', {'A_Row_Index', 'Matrix_Xdrc_tr', 'B_Row_Index', 'Matrix_Xdrc'});
	ResultTable.B_Row_Index(~tf) = NaN;
	disp(ResultTable);
	
	%% Verification 4: equivalence of the subspaces of Udr and Udrc
	disp(' ')
	disp('Equivalence of the subspaces of Udr and Udrc')
	T = Udrc\Udr;
	Ydrc_tr = Ydr/T;
	[tf, loc] = ismember(round(Ydrc_tr,10), round(Ydrc,10), 'rows');
	Ydrc_ext = NaN(size(Ydrc_tr));
	Ydrc_ext(logical(loc),:) = Ydrc;
	ResultTable = table((1:size(Ydrc_tr,1))', Ydrc_tr, loc, Ydrc_ext, ...
	'VariableNames', {'A_Row_Index', 'Matrix_Ydrc_tr', 'B_Row_Index', 'Matrix_Ydrc'});
	ResultTable.B_Row_Index(~tf) = NaN;
	disp(ResultTable);
\end{lstlisting}
\end{strip}

\subsubsection*{Helper function: \texttt{sign\_flip}}

\begin{strip}
\begin{lstlisting}[style=matlabstyle, caption={Sign-flip normalization for SVD/PCA loadings (Bro, Acar \& Kolda, 2007)}, label={lst:sign-flip}]
	function [sgns,loads] = sign_flip(loads,X)
	% [sgns,loads] = sign_flip(loads,X)
	% loads is a cell of loadings
	% X     is the data array
	% sgns is a MxF matrix where sgns(m,f) is the sign of loading f in mode m
	%
	% If using svd ([u,s,v]=svd(X)) then loads{1}=u*s; and loads{2}=v;
	% If using an F-component PCA model ([t,p]=pca(X,F), then loads{1}=t; and
	% loads{2}=p;
	%
	% Copyright 2007 R. Bro, E. Acar, T. Kolda - www.models.life.ku.dk
	% for PARAFAC and two-way
	% https://doi.org/10.1002/cem.1122
	
	if isa(X,'dataset')
	inc = X.includ;
	X = X.data(inc{:});
	end
	
	order = length(size(X));
	for i=1:order
	F(i) = size(loads{i},2);
	end
	
	for m = 1:order % for each mode
	for f=1:F(m) % for each component
	s=[];
	a = loads{m}(:,f);
	a = a /(a'*a);
	x = subtract_otherfactors(X, loads, m, f);
	for i=1:size(x(:,:),2) % for each column
	s(i)=(a'*x(:,i));
	s(i)=sign(s(i))*power(s(i),2);
	end
	S(m,f) = sum(s);
	end
	end
	
	sgns = sign(S);
	for f=1:F(1) % each component
	for i=1:size(sgns,1) % each mode
	se = length(find(sgns(:,f)==-1));
	if (rem(se,2)==0 )
	loads{i}(:,f) = sgns(i,f)*loads{i}(:,f);
	else
	disp('Odd number of negatives!')
	sgns(:,f) = handle_oddnumbers(S(:,f));
	se = length(find(sgns(:,f)==-1));
	if (rem(se,2)==0)
	loads{i}(:,f) = sgns(i,f)*loads{i}(:,f);
	else
	disp('Something Wrong!!!')
	end
	end
	end % each mode
	end % each component
	end
\end{lstlisting}
\end{strip}

\subsubsection*{Helper function: \texttt{handle\_oddnumbers}}

\begin{strip}
\begin{lstlisting}[style=matlabstyle, caption={Tie-breaking rule when the sign vote is ambiguous}, label={lst:handle-odd}]
	function sgns = handle_oddnumbers(Bcon)
	sgns = sign(Bcon);
	nb_neg = find(Bcon<0);
	[~, index] = min(abs(Bcon));
	if (Bcon(index)<0)
	sgns(index) = -sgns(index);
	% since this function is called nb_neg should be greater than 0, anyway
	elseif ((Bcon(index)>0) && (nb_neg>0))
	sgns(index) = -sgns(index);
	end
	end
\end{lstlisting}
\end{strip}

\subsubsection*{Helper function: \texttt{subtract\_otherfactors} and \texttt{outerm}}

\begin{strip}
\begin{lstlisting}[style=matlabstyle, caption={Residual computation and outer-product helper for the sign-flip routine}, label={lst:subtract-outerm}]
	function x = subtract_otherfactors(X, loads, mode, factor)
	order = length(size(X));
	x = permute(X,[mode 1:mode-1 mode+1:order]);
	loads = loads([mode 1:mode-1 mode+1:order]);
	for m = 1:order
	loads{m} = loads{m}(:, [factor 1:factor-1 factor+1:size(loads{m},2)]);
	L{m} = loads{m}(:,2:end);
	end
	M = outerm(L);
	x = x - M;
	end
	
	function mwa = outerm(facts,lo,vect)
	if nargin < 2
	lo = 0;
	end
	if nargin < 3
	vect = 0;
	end
	order = length(facts);
	if lo == 0
	mwasize = zeros(1,order);
	else
	mwasize = zeros(1,order-1);
	end
	k = 0;
	for i = 1:order
	if i ~= lo
	[m,n] = size(facts{i});
	k = k + 1;
	mwasize(k) = m;
	if k > 1
	else
	nofac = n;
	end
	end
	end
	mwa = zeros(prod(mwasize),nofac);
	for j = 1:nofac
	if lo ~= 1
	mwvect = facts{1}(:,j);
	for i = 2:order
	if lo ~= i
	mwvect = mwvect*facts{i}(:,j)';
	mwvect = mwvect(:);
	end
	end
	elseif lo == 1
	mwvect = facts{2}(:,j);
	for i = 3:order
	mwvect = mwvect*facts{i}(:,j)';
	mwvect = mwvect(:);
	end
	end
	mwa(:,j) = mwvect;
	end
	% If vect isn't one, sum up the results of the factors and reshape
	if vect ~= 1
	mwa = sum(mwa,2);
	mwa = reshape(mwa,mwasize);
	end
	end
\end{lstlisting}
\end{strip}

\subsection{MATLAB Code: Essential Data Points (EDPs) Strict Verification}\label{App:Sec-4-fundSubSpStrct}
\label{app:matlab-edp-strict}

The following script numerically verifies how the four fundamental subspaces and SVD factors of a rank-$r$ matrix $R$ change when rows and/or columns that are not convex-hull vertices of the corresponding point clouds are deleted. It supports two datasets: a random synthetic rank-$r$ matrix (\texttt{whdat=1}), and the real Henry \& Kim (1990) source-apportionment dataset~\cite{HENRY1990} (\texttt{whdat=2}), for which the convex-hull index sets were pre-computed. All derived identities (pure rotations, ambient-shrinking isomorphisms, $\Sigma$-mismatch isomorphisms, and the double-deletion composition) are checked to machine precision.

\begin{strip}
\begin{lstlisting}[style=matlabstyle, caption={Convex-hull SVD subspace verification script}, label={lst:svd-subspaces}]
	%% convex_hull_svd_subspaces.mlx
	%
	% Illustrates, numerically, how the four fundamental subspaces / SVD
	% factors of a rank-r matrix R change when we delete rows and/or columns
	% that are NOT convex-hull vertices of the corresponding point clouds.
	%
	%   R    = X *V'    ,  X  = U *S  ,   Y  = V *S       (full data)
	%   Rdr  = Xdr*Vdr' ,  Xdr = Udr*Sdr, Ydr = Vdr*Sdr   (hull-vertex ROWS kept)
	%   Rdc  = Xdc*Vdc' ,  Xdc = Udc*Sdc, Ydc = Vdc*Sdc   (hull-vertex COLUMNS kept)
	%   Rdrc = Xdrc*Vdrc', Xdrc= Udrc*Sdrc,Ydrc= Vdrc*Sdrc (both)
	%
	% All the derived identities are checked to machine precision.
	
	clear; close all; clc; rng(1);
	
	%% 1. Build a random rank-r matrix and its compact SVD
	whdat = 2;
	switch whdat
	case 1
	m = 20; n = 18; r = 3;
	A = randn(m,r); B = randn(r,n);
	R = A*B;                                  % rank-r data matrix
	case 2
	% Source compositions
	% Marine    Udust     Auto
	A = [[ 0.40000  0.01250  0.00000]; ...   % Na
	[ 0.00000  0.08840  0.01100]; ...   % Al
	[ 0.00000  0.22300  0.00820]; ...   % Si
	[ 0.40000  0.00000  0.03000]; ...   % Cl
	[ 0.01400  0.01030  0.00072]; ...   % K
	[ 0.01400  0.02440  0.01250]; ...   % Ca
	[ 0.00000  0.00640  0.00000]; ...   % Ti
	[ 0.00000  0.06000  0.02100]; ...   % Fe
	[ 0.00200  0.00020  0.05000]; ...   % Br
	[ 0.00000  0.00370  0.20000]];      % Pb
	
	% Source apportionment ug/m3
	% Marine    Udust     Auto
	B = [[   3         8        19]; ...
	[   3         8        16]; ...
	[   5         7        12]; ...
	[   6        49        13]; ...
	[   6        39         7]; ...
	[   6         9        12]; ...
	[   6        17        18]; ...
	[   2         6         6]; ...
	[   4        42         7]; ...
	[   4        49        20]; ...
	[   4        29        14]; ...
	[   5        44         9]; ...
	[   5        32        12]; ...
	[   5         9        11]; ...
	[   6        26         8]; ...
	[   3         8        14]; ...
	[   2        39        11]; ...
	[   3         6        20]; ...
	[   4        48        15]; ...
	[   5        37        14]];
	
	R = A*B';
	m = size(R,1); n = size(R,2); r = size(A,2);
	end
	
	[U,S,V] = svds(R,r);                      % compact SVD: R = U*S*V'
	X = U*S;                                  % "row scores"    (m x r)
	Y = V*S;                                  % "column scores" (n x r)
	assert(norm(R-X*V','fro')<1e-9)
	assert(norm(R-U*Y','fro')<1e-9)
	
	%% 2. Find convex-hull vertices of the row-points (X) and column-points (Y)
	switch whdat
	case 1
	I = sort(unique(convhulln(X)));           % row indices to KEEP
	J = sort(unique(convhulln(Y)));           % column indices to KEEP
	case 2
	IndsH3.XnInn = [1 4; 4 9; 9 10; 10 7; 7 1];
	I = unique(IndsH3.XnInn);
	IndsH3.YnInn = [5 15; 15 6; 6 3; 3 18; 18 17; 17 9; 9 5];
	J = unique(IndsH3.YnInn);
	% I = sort(unique(convhulln(X)));         % row indices to KEEP
	% J = sort(unique(convhulln(Y)));         % column indices to KEEP
	end
	fprintf('Kept rows:    %d / %d\n', numel(I), m);
	fprintf('Kept columns: %d / %d\n', numel(J), n);
	
	%% 3. Build the three reduced matrices and their compact SVDs
	Rdr  = R(I,:);
	Rdc  = R(:,J);
	Rdrc = R(I,J);
	
	[Udr ,Sdr ,Vdr ] = svds(Rdr ,r);
	[Udc ,Sdc ,Vdc ] = svds(Rdc ,r);
	[Udrc,Sdrc,Vdrc] = svds(Rdrc,r);
	
	Xdr  = Udr*Sdr;   Ydr  = Vdr*Sdr;
	Xdc  = Udc*Sdc;   Ydc  = Vdc*Sdc;
	Xdrc = Udrc*Sdrc; Ydrc = Vdrc*Sdrc;
	
	fprintf('\nRank check (must all equal r=%d): %d %d %d\n', r, ...
	rank(Rdr), rank(Rdc), rank(Rdrc));
	
	%% TYPE 1 -- pure rotations (subspace identical, ambient space unchanged)
	%    Row space of R is exactly preserved under row deletion  -> V, Vdr
	%    Column space of R is exactly preserved under column del.-> U, Udc
	Qdr = V'*Vdr;     % should be orthogonal r x r
	Qdc = U'*Udc;     % should be orthogonal r x r
	
	fprintf('\n--- TYPE 1: pure rotations ---\n');
	fprintf('||Qdr''Qdr - I||   = %.3e (orthogonality)\n', norm(Qdr'*Qdr-eye(r)));
	fprintf('||V*Qdr  - Vdr||   = %.3e (row space of R preserved)\n', norm(V*Qdr-Vdr,'fro'));
	fprintf('||Qdc''Qdc - I||   = %.3e (orthogonality)\n', norm(Qdc'*Qdc-eye(r)));
	fprintf('||U*Qdc  - Udc||   = %.3e (col space of R preserved)\n', norm(U*Qdc-Udc,'fro'));
	
	%% TYPE 2 -- ambient-shrinking isomorphisms (general invertible, not orthogonal)
	%    U restricted to kept rows I is an injective image of U_dr (m -> m')
	%    V restricted to kept cols J is an injective image of V_dc (n -> n')
	Ndr = Udr \ U(I,:);     % U(I,:) = Udr*Ndr
	Mco = Vdc \ V(J,:);     % V(J,:) = Vdc*Mco
	
	fprintf('\n--- TYPE 2: ambient-shrinking isomorphisms ---\n');
	fprintf('||U(I,:) - Udr*Ndr||  = %.3e\n', norm(U(I,:)-Udr*Ndr,'fro'));
	fprintf('||V(J,:) - Vdc*Mco||  = %.3e\n', norm(V(J,:)-Vdc*Mco,'fro'));
	fprintf('Ndr orthogonal? ||Ndr''Ndr-I|| = %.3e (generally NOT 0)\n', norm(Ndr'*Ndr-eye(r)));
	
	%% Consequence for X and Y -- clean rotation on the RESTRICTED factor
	%    Because X = U*S and V_dr = V*Qdr, the S-factors cancel exactly:
	%       Xdr = X(I,:)*Qdr        Ydc = Y(J,:)*Qdc
	fprintf('\n--- Derived clean relations for X, Y ---\n');
	fprintf('||Xdr - X(I,:)*Qdr||  = %.3e\n', norm(Xdr - X(I,:)*Qdr,'fro'));
	fprintf('||Ydc - Y(J,:)*Qdc||  = %.3e\n', norm(Ydc - Y(J,:)*Qdc,'fro'));
	
	%% TYPE 3 -- Sigma-mismatch isomorphisms (same ambient space, values differ)
	%    X (full, m rows) vs Xdc: ambient R^m unchanged, but Sigma != Sigma_dc
	%    Y (full, n rows) vs Ydr: ambient R^n unchanged, but Sigma != Sigma_dr
	Kdc = (X'*X) \ (X'*Xdc);        % Xdc = X*Kdc
	Ldr = (Y'*Y) \ (Y'*Ydr);        % Ydr = Y*Ldr
	Kdc_theory = inv(S)*Qdc*Sdc;
	Ldr_theory = inv(S)*Qdr*Sdr;
	
	fprintf('\n--- TYPE 3: Sigma-mismatch isomorphisms ---\n');
	fprintf('||Xdc - X*Kdc||             = %.3e\n', norm(Xdc-X*Kdc,'fro'));
	fprintf('||Kdc - S^{-1}*Qdc*Sdc||    = %.3e\n', norm(Kdc-Kdc_theory,'fro'));
	fprintf('||Ydr - Y*Ldr||             = %.3e\n', norm(Ydr-Y*Ldr,'fro'));
	fprintf('||Ldr - S^{-1}*Qdr*Sdr||    = %.3e\n', norm(Ldr-Ldr_theory,'fro'));
	
	%% Double deletion (drc): composition of the above, and path-independence
	Qdrc = Vdc'*Vdrc;                % row space of R(:,J) preserved by further row-deletion
	Xdrc_theory = X(I,:)*Mco'*Qdrc;  % composed transform (col-deletion isomorphism, then row rotation)
	
	fprintf('\n--- Double deletion (drc) ---\n');
	fprintf('||Qdrc''Qdrc - I||            = %.3e (orthogonality)\n', norm(Qdrc'*Qdrc-eye(r)));
	fprintf('||Vdc*Qdrc - Vdrc||           = %.3e (subspace preserved)\n', norm(Vdc*Qdrc-Vdrc,'fro'));
	fprintf('||Xdrc - X(I,:)*Mco''*Qdrc||   = %.3e (path: rows kept -> cols isomorphism -> rotation)\n', ...
	norm(Xdrc - Xdrc_theory,'fro'));
	
	%% Summary
	fprintf(['\nSUMMARY\n' ...
	' V,  Vdr : IDENTICAL subspace of R^n   (rotation Qdr)\n' ...
	' U,  Udc : IDENTICAL subspace of R^m   (rotation Qdc)\n' ...
	' U(I,:), Udr : ISOMORPHIC (ambient shrinks m->m'', map Ndr)\n' ...
	' V(J,:), Vdc : ISOMORPHIC (ambient shrinks n->n'', map Mco)\n' ...
	' X,  Xdc : related by general invertible Kdc (Sigma changes, ambient fixed)\n' ...
	' Y,  Ydr : related by general invertible Ldr (Sigma changes, ambient fixed)\n' ...
	' Xdr = X(I,:)*Qdr   and   Ydc = Y(J,:)*Qdc   (clean, exact, always)\n' ...
	' Vdrc, Xdrc obtained by composing the column-isomorphism with a further rotation\n']);
\end{lstlisting}
\end{strip}

\section*{Acknowledgements}

Work was partially supported by the Distinguished Professor Program of Óbuda University. R. Rajk\'o is also grateful for the opportunity to use the HUN-REN Cloud \url{https://science-cloud.hu/en} (accessed on 28 August 2026)~\cite{H_der_2022} which helped achieve some specific results published in this paper.

\printbibliography

@article{HENRY1990,
	title = {Extension of self-modeling curve resolution to mixtures of more than three components: Part 1. Finding the basic feasible region},
	journal = {Chemometrics and Intelligent Laboratory Systems},
	volume = {8},
	number = {2},
	pages = {205-216},
	year = {1990},
	issn = {0169-7439},
	doi = {https://doi.org/10.1016/0169-7439(90)80136-T},
	author = {Ronald C. Henry and Bong Mann Kim}
}

@article{Pearson1901,
  author    = {Karl Pearson},
  title     = {{LIII. On lines and planes of closest fit to systems of points in space}},
  journal   = {The London, Edinburgh, and Dublin Philosophical Magazine and Journal of Science},
  volume    = {2},
  number    = {11},
  pages     = {559--572},
  year      = {1901},
  publisher = {Taylor & Francis},
  doi       = {10.1080/14786440109462720}
}

@article{Scholkopf1998,
  author    = {Sch{\"o}lkopf, Bernhard and Smola, Alexander and M{\"u}ller, Klaus-Robert},
  title     = {Nonlinear Component Analysis as a Kernel Eigenvalue Problem},
  journal   = {Neural Computation},
  volume    = {10},
  number    = {5},
  pages     = {1299--1319},
  year      = {1998},
  publisher = {MIT Press},
  doi       = {10.1162/089976698300017467}
}

@article{vandermaaten2008,
  author  = {Laurens van der Maaten and Geoffrey Hinton},
  title   = {Visualizing Data using t-SNE},
  journal = {Journal of Machine Learning Research},
  year    = {2008},
  volume  = {9},
  number  = {86},
  pages   = {2579--2605},
  url     = {http://www.jmlr.org/papers/v9/vandermaaten08a.html}
}

@article{ayesha2020,
  title={Overview and comparative study of dimensionality reduction techniques for high dimensional data},
  author={Ayesha, Shaeela and Hanif, Muhammad Kashif and Talib, Ramzan},
  journal={Information Fusion},
  volume={59},
  pages={44--58},
  year={2020},
  publisher={Elsevier},
  doi={10.1016/j.inffus.2020.01.005}
}

@article{wani2025,
  title={Comprehensive review of dimensionality reduction algorithms: challenges, limitations, and innovative solutions},
  author={Wani, Aasim Ayaz},
  journal={PeerJ Computer Science},
  volume={11},
  pages={e3025},
  year={2025},
  month={Jul},
  publisher={PeerJ Inc.},
  doi={10.7717/peerj-cs.3025},
  url={https://doi.org/10.7717/peerj-cs.3025}
}

@article{KHODADADIKARIMVAND2023,
title = {Practical and comparative application of efficient data reduction - Multivariate curve resolution},
journal = {Analytica Chimica Acta},
volume = {1243},
pages = {340824},
year = {2023},
issn = {0003-2670},
doi = {https://doi.org/10.1016/j.aca.2023.340824},
url = {https://www.sciencedirect.com/science/article/pii/S0003267023000454},
author = {Somaiyeh {Khodadadi Karimvand} and Jamile {Mohammad Jafari} and Somaye {Vali Zade} and Hamid Abdollahi},
}

@article{cordina2026,
  author    = {Cordina, Robert and Bates, Matthew},
  title     = {Interpretable Dimension Reduction and Clustering for Complex Chromatographic Datasets},
  journal   = {LCGC International -- Chromatography Online},
  year      = {2026},
  month     = {September},
  day       = {2},
  url       = {https://www.chromatographyonline.com/view/interpretable-dimension-reduction-and-clustering-for-complex-chromatographic-datasets},
  note      = {Accessed: 2026-09-07}
}

@article{GARISO2025,
title = {A comparative analysis of deep learning and chemometric approaches for spectral data modeling},
journal = {Analytica Chimica Acta},
volume = {1347},
pages = {343766},
year = {2025},
issn = {0003-2670},
doi = {https://doi.org/10.1016/j.aca.2025.343766},
url = {https://www.sciencedirect.com/science/article/pii/S0003267025001606},
author = {Rúben Gariso and João P.L. Coutinho and Tiago J. Rato and Marco S. Reis},
}

@article{Ghaffari2019,
  author = {Ghaffari, Mehdi and Omidikia, Narges and Ruckebusch, Cyril},
  title = {Essential Spectral Pixels for Multivariate Curve Resolution of Chemical Images},
  journal = {Analytical Chemistry},
  volume = {91},
  number = {17},
  pages = {10943--10948},
  year = {2019},
  doi = {10.1021/acs.analchem.9b02890},
  url = {https://doi.org}
}

@article{Ghaffari2021,
  title   = {Joint selection of essential pixels and essential variables across hyperspectral images},
  author  = {Mahdiyeh Ghaffari and Nematollah Omidikia and Cyril Ruckebusch},
  journal = {Analytica Chimica Acta},
  volume  = {1141},
  pages   = {36--46},
  year    = {2021},
  issn    = {0003-2670},
  doi     = {10.1016/j.aca.2020.10.040}
}

@article{Ruckebusch2020,
  title     = {Perspective on essential information in multivariate curve resolution},
  author    = {Ruckebusch, C. and Vitale, R. and Ghaffari, M. and Hugelier, S. and Omidikia, N.},
  journal   = {TrAC Trends in Analytical Chemistry},
  volume    = {132},
  pages     = {116044},
  year      = {2020},
  issn      = {0165-9936},
  doi       = {10.1016/j.trac.2020.116044}
}

@article{Fernandez2002,
  title = {Methods for outlier detection in prediction},
  journal = {Chemometrics and Intelligent Laboratory Systems},
  volume = {63},
  number = {1},
  pages = {27--39},
  year = {2002},
  issn = {0169-7439},
  doi = {https://doi.org},
  author = {J.A. Fernández Pierna and F. Wahl and O.E. de Noord and D.L. Massart}
}

@article{fernandez2003,
  author     = {Fern{\'a}ndez Pierna, J. A. and Jin, L. and Daszykowski, M. and Wahl, F. and Massart, D. L.},
  title      = {A methodology to detect outliers/inliers in prediction with PLS},
  journal    = {Chemometrics and Intelligent Laboratory Systems},
  volume     = {68},
  number     = {1-2},
  pages      = {17--28},
  year       = {2003},
  publisher = {Elsevier},
  doi        = {10.1016/S0169-7439(03)000844}
}

@article{Jin2003,
  title = {The Law of Mixtures method for multivariate calibration},
  journal = {Analytica Chimica Acta},
  volume = {477},
  number = {1},
  pages = {3-15},
  year = {2003},
  issn = {0003-2670},
  doi = {https://doi.org},
  url = {https://sciencedirect.com},
  author = {L. Jin and J.A. Fernández Pierna and F. Wahl and P. Dardenne and D.L. Massart}
}

@article{Jin2003Delaunay,
  title = {Delaunay triangulation method for multivariate calibration},
  journal = {Analytica Chimica Acta},
  volume = {489},
  number = {2},
  pages = {195-206},
  year = {2003},
  issn = {0003-2670},
  doi = {10.1016/S0003-2670(03)00629-9},
  author = {L. Jin and J.A. Fernández Pierna and Q. Xu and F. Wahl and O.E. de Noord and C.A. Saby and D.L. Massart}
}

@article{JIN2006,
  title = {Updating multivariate calibration with the Delaunay triangulation method: The creation of a new local model},
  journal = {Chemometrics and Intelligent Laboratory Systems},
  volume = {80},
  number = {1},
  pages = {124-135},
  year = {2006},
  issn = {0169-7439},
  doi = {https://doi.org},
  author = {L. Jin and Q.S. Xu and J. Smeyers-Verbeke and D.L. Massart}
}

@manual{corona2010,
  author    = {Corona, Francesco and Liiti{"{a}}inen, Elia and Lendasse, Amaury and Baratti, Roberto and Sassu, Lorenzo},
  title     = {A continuous regression function for the {D}elaunay calibration method},
  journal   = {IFAC Proceedings Volumes},
  volume    = {43},
  number    = {5},
  pages     = {201--206},
  year      = {2010},
  doi       = {10.3182/20100607-3-IN-4014.00036},
  url       = {https://www.sciencedirect.com/science/article/pii/S1474667016300337},
  publisher = {Elsevier}
}

@article{Carlosena1995,
    author = {Carlosena, A. and Andrade, J. M. and Kubista, M. and Prada, D.},
    title = {Procrustes Rotation as a Way To Compare Different Sampling Seasons in Soils},
    journal = {Analytical Chemistry},
    volume = {67},
    number = {14},
    pages = {2373-2378},
    year = {1995},
    month = {05},
    issn = {0003-2700},
    doi = {10.1021/ac00110a008},
    url = {https://doi.org/10.1021/ac00110a008}
}

@article{ANDRADE2004,
title = {Procrustes rotation in analytical chemistry, a tutorial},
journal = {Chemometrics and Intelligent Laboratory Systems},
volume = {72},
number = {2},
pages = {123-132},
year = {2004},
issn = {0169-7439},
doi = {https://doi.org/10.1016/j.chemolab.2004.01.007},
url = {https://www.sciencedirect.com/science/article/pii/S0169743904000152},
author = {Jose Manuel Andrade and Marı́a P. Gómez-Carracedo and Wojtek Krzanowski and Mikael Kubista}
}

@article{Ivosev2008,
    author = {Ivosev, Gordana and Burton, Lyle and Bonner, Ron},
    title = {Dimensionality Reduction and Visualization in Principal Component Analysis},
    journal = {Analytical Chemistry},
    volume = {80},
    number = {13},
    pages = {4933-4944},
    year = {2008},
    month = {06},
    issn = {0003-2700},
    doi = {10.1021/ac800110w}
}

@article{GONCALVES2023,
title = {Exploring the scores: {P}rocrustes analysis for comprehensive exploration of multivariate data},
journal = {Chemometrics and Intelligent Laboratory Systems},
volume = {238},
pages = {104841},
year = {2023},
issn = {0169-7439},
doi = {https://doi.org/10.1016/j.chemolab.2023.104841},
url = {https://www.sciencedirect.com/science/article/pii/S0169743923000916},
author = {Thays R. Gonçalves and Peter D. Wentzell and Makoto Matsushita and Patrícia Valderrama}
}

@article{andreella2022,
  author    = {Andreella, Angela and Finos, Livio},
  title     = {Procrustes Analysis for High-Dimensional Data},
  journal   = {Psychometrika},
  volume    = {87},
  number    = {4},
  pages     = {1422--1438},
  year      = {2022},
  publisher = {Springer},
  doi       = {10.1007/s11336-022-09859-5},
  url       = {https://doi.org/10.1007/s11336-022-09859-5}
}

@article{kucheryavskiy2020,
  author    = {Kucheryavskiy, Sergey and Zhilin, Sergei and Rodionova, Oxana and Pomerantsev, Alexey},
  title     = {Procrustes Cross-Validation—A Bridge between Cross-Validation and Independent Validation Sets},
  journal   = {Analytical Chemistry},
  volume    = {92},
  number    = {17},
  pages     = {11842--11850},
  year      = {2020},
  doi       = {10.1021/acs.analchem.0c02175},
  publisher = {ACS Publications}
}

@article{rajko2005,
author = {Rajk\'{o}, R\'{o}bert and Istv\'{a}n, Krisztina},
title = {Analytical solution for determining feasible regions of self-modeling curve resolution (SMCR) method based on computational geometry},
journal = {Journal of Chemometrics},
volume = {19},
number = {8},
pages = {448-463},
doi = {https://doi.org/10.1002/cem.947},
url = {https://analyticalsciencejournals.onlinelibrary.wiley.com/doi/abs/10.1002/cem.947},
eprint = {https://analyticalsciencejournals.onlinelibrary.wiley.com/doi/pdf/10.1002/cem.947},
year = {2005}
}

@article{BORGEN1985,
title = {An extension of the multivariate component-resolution method to three components},
journal = {Analytica Chimica Acta},
volume = {174},
pages = {1-26},
year = {1985},
issn = {0003-2670},
doi = {https://doi.org/10.1016/S0003-2670(00)84361-5},
url = {https://www.sciencedirect.com/science/article/pii/S0003267000843615},
author = {Odd S. Borgen and Bruce R. Kowalski}
}

@misc{rajko2017,
  author       = {Rajk\'{o}, R\'{o}bert},
  title        = {Use and abuse of curve resolution in chemometrics},
  howpublished = {Invited lecture (45 min.) \url{https://m2.mtmt.hu/gui2/?mode=browse&params=publication;3283645}},
  year         = {},
  month        = {April},
  address      = {Newcastle, Australia},
  note         = {Topics in Chemometrics TIC2017, 18-21 April 2017},
}

@book{strang2023,
  author    = {Strang, Gilbert},
  title     = {Introduction to Linear Algebra},
  edition   = {6th},
  year      = {2023},
  publisher = {Wellesley-Cambridge Press},
  address   = {Wellesley, MA},
  isbn      = {978-1-7331466-7-8}
}

@article{OLARINI2024,
title = {Exploratory analysis of hyperspectral imaging data},
journal = {Chemometrics and Intelligent Laboratory Systems},
volume = {252},
pages = {105174},
year = {2024},
issn = {0169-7439},
doi = {https://doi.org/10.1016/j.chemolab.2024.105174},
author = {Alessandra Olarini and Marina Cocchi and Vincent Motto-Ros and Ludovic Duponchel and Cyril Ruckebusch}
}

@article{QING2024,
title = {Essential spectral pixels-based improvement of UMAP classifying hyperspectral imaging data to identify minor compounds in food matrix},
journal = {Talanta},
volume = {273},
pages = {125845},
year = {2024},
issn = {0039-9140},
doi = {https://doi.org/10.1016/j.talanta.2024.125845},
author = {Xiangdong Qing and Guiying Lu and Xiaohua Zhang and Qingling Chen and Xiaohong Zhou and Wei He and Ling Xu and Jin Zhang}
}

@article{Coic2023,
    author = {Coic, Laureen and Vitale, Raffaele and Moreau, Myriam and Rousseau, David and de Morais Goulart, José Henrique and Dobigeon, Nicolas and Ruckebusch, Cyril},
    title = {Assessment of Essential Information in the Fourier
Domain to Accelerate Raman Hyperspectral Microimaging},
    journal = {Analytical Chemistry},
    volume = {95},
    number = {42},
    pages = {15497-15504},
    year = {2023},
    month = {10},
    issn = {0003-2700},
    doi = {10.1021/acs.analchem.3c01383}
}

@article{Beyramysoltan2021,
    author = {Beyramysoltan, Samira and Abdollahi, Hamid and Musah, Rabi A.},
    title = {Workflow for the Supervised Learning of Chemical Data:
Efficient Data Reduction-Multivariate Curve Resolution (EDR-MCR)},
    journal = {Analytical Chemistry},
    volume = {93},
    number = {12},
    pages = {5020-5027},
    year = {2021},
    month = {03},
    issn = {0003-2700},
    doi = {10.1021/acs.analchem.0c01427}
}

@article{Cheng2005,
author = {Cheng, H. and Gimbutas, Z. and Martinsson, P. G. and Rokhlin, V.},
title = {On the Compression of Low Rank Matrices},
journal = {SIAM Journal on Scientific Computing},
volume = {26},
number = {4},
pages = {1389-1404},
year = {2005},
doi = {10.1137/030602678}
}

@article{Mahoney2009,
author = {Michael W. Mahoney  and Petros Drineas },
title = {CUR matrix decompositions for improved data analysis},
journal = {Proceedings of the National Academy of Sciences},
volume = {106},
number = {3},
pages = {697-702},
year = {2009},
doi = {10.1073/pnas.0803205106}
}

@article{Cortinovis2020,
author = {Cortinovis, Alice and Kressner, Daniel},
title = {Low-Rank Approximation in the Frobenius Norm by Column and Row Subset Selection},
journal = {SIAM Journal on Matrix Analysis and Applications},
volume = {41},
number = {4},
pages = {1651-1673},
year = {2020},
doi = {10.1137/19M1281848}
}

@ARTICLE{Ding2018,
  author={Ding, Shuguang and Nie, Xiangli and Qiao, Hong and Zhang, Bo},
  journal={IEEE Transactions on Neural Networks and Learning Systems}, 
  title={A Fast Algorithm of Convex Hull Vertices Selection for Online Classification}, 
  year={2018},
  volume={29},
  number={4},
  pages={792-806},
  doi={10.1109/TNNLS.2017.2648038}
}

@book{ROCKAFELLAR1970,
 ISBN = {9780691015866},
 URL = {http://www.jstor.org/stable/j.ctt14bs1ff},
 author = {R. Tyrrell Rockafellar},
 publisher = {Princeton University Press},
 title = {Convex Analysis},
 urldate = {2026-09-14},
 year = {1970}
}

@article{Vitale2024,
author = {Vitale, Raffaele and Azizi, Azar and Ghaffari, Mahdiyeh and Omidikia, Nematollah and Ruckebusch, Cyril},
title = {Three-Way Data Reduction Based on Essential Information},
journal = {Journal of Chemometrics},
volume = {38},
number = {12},
pages = {e3617},
doi = {https://doi.org/10.1002/cem.3617},
year = {2024}
}

@article{H_der_2022,
	doi = {10.22503/inftars.xxii.2022.2.8},
	url = {https://doi.org/10.22503/inftars.XXII.2022.2.8},
	year = 2022,
	month = {aug},
	publisher = {Informacios Tarsadalom},
	volume = {22},
	number = {2},
	pages = {128-137},
	author = {Mih{\'{a}}ly H{\'{e}}der and Ern{\H{o}} Rig{\'{o}} and Dorottya Medgyesi and R{\'{o}}bert Lovas and Szabolcs Tenczer and Ferenc T{\"o}r{\"o}k and Attila Farkas and M{\'{a}}rk Em{\H{o}}di and J{\'{o}}zsef Kadlecsik and Gy{\"o}rgy Mez{\H{o}} and {\'{A}}d{\'{a}}m Pint{\'{e}}r and P{\'{e}}ter Kacsuk},
	title = {The Past, Present and Future of the {ELKH} Cloud},
	journal = {Inform{\'{a}}ci{\'{o}}s T{\'{a}}rsadalom}
}

@article{AZZOUZ2008,
title = {Application of multivariate curve resolution alternating least squares (MCR-ALS) to the quantitative analysis of pharmaceutical and agricultural samples},
journal = {Talanta},
volume = {74},
number = {5},
pages = {1201-1210},
year = {2008},
issn = {0039-9140},
doi = {https://doi.org/10.1016/j.talanta.2007.08.024},
author = {T. Azzouz and R. Tauler},
}

@article{RAJKO2024,
title = {On problematic practice of using normalization in self-modeling/multivariate curve resolution (S/MCR)},
journal = {Chemometrics and Intelligent Laboratory Systems},
volume = {244},
pages = {105033},
year = {2024},
issn = {0169-7439},
doi = {https://doi.org/10.1016/j.chemolab.2023.105033},
author = {Róbert Rajkó}
}


\end{document}